\documentclass[12pt,draftcls, onecolumn]{IEEEtran}
\usepackage{amsmath,amssymb,amsfonts,mathrsfs}
\usepackage{epsfig}
\usepackage{graphicx,graphics}
\usepackage{color}
\usepackage{hyperref}

\newtheorem{theorem}{Theorem}

\newtheorem{definition}{Definition}
\newtheorem{lemma}{Lemma}
\newtheorem{proposition}{Proposition}
\newtheorem{remark}{Remark}

\begin{document}

\title{On the response of quantum linear systems to single photon input fields}

\author{Guofeng~Zhang\thanks{%
G.~Zhang is with the Department of Applied Mathematics, The Hong Kong Polytechnic University, Hong Kong. He was with the Research School of Engineering, The Australian National University, Canberra, ACT 0200, Australia.  (e-mail: Guofeng.Zhang@polyu.edu.hk).} ~~ Matthew~R.~James\thanks{M.~James is with the ARC Centre for Quantum Computation and Communication Technology, Research  School of Engineering, Australian National University, Canberra, ACT 0200, Australia. (e-mail: Matthew.James@anu.edu.au).}
}
\maketitle

\begin{abstract}

The purpose of this paper is to extend linear systems and signals theory to include single photon quantum signals. We provide detailed results describing how quantum linear systems respond to multichannel single photon  quantum signals. In particular,  we characterize the class of states (which we call {\em photon-Gaussian} states) that result when multichannel photons are input to a quantum linear system. We show that this class of quantum states is preserved by quantum linear systems. Multichannel photon-Gaussian states are defined via the action of certain creation and annihilation operators on Gaussian states.  Our results show how the output states are determined from the input states through a pair of transfer function relations. We also provide  equations from which output signal intensities can be computed. Examples from quantum optics are provided to illustrate the results.

\end{abstract}

\tableofcontents

\section{Introduction}
\label{sec:introduction}

One of the most basic aspects of systems and control theory is the study of how systems respond to input signals.  Well known tools, including transfer functions and impulse response functions, allow engineers to determine  the output signal produced by a linear system in response to a given input signal.  Such knowledge is needed for engineers to enable them to analyze and design control systems.

As is well known, signals other than deterministic signals are important to a wide range of applications. There is a well-developed theory of dynamical systems and non-deterministic signals, within which  Gaussian signals play an important role. Indeed,
linear systems (with Gaussian states) driven by Gaussian  input signals  provide the foundations for Kalman filtering and linear quadratic Gaussian (LQG) control, as well as many other developments and applications, \cite{AM71,KS72,MHAD77,AM79,KV86}.
It is well known that for a classical linear system $G$ initialized in a Gaussian state and driven by Gaussian white noise process $w(t)$ satisfying $\mathbb{E}[w(t)]=0$ and $\mathbb{E}[w(t)w(t^\prime)^T]=\delta(t-t^\prime)$,
\begin{eqnarray}
 \dot{x}(t) &=& Ax(t)+Bw(t) , \ \ x(t_0) = x,
 \label{eq:lcss} \\
  y(t) &=& Cx(t) + D w(t), \nonumber
\end{eqnarray}
where $\mathbb{E}$ denote mathematical expectation, the mean $\bar x(t)=\mathbb{E}[x(t)]$ and covariance  $\Sigma (t) = \mathbb{E} \left[(x(t)-\bar{x}(t))(x(t)-\bar{x}(t))^T\right]$ satisfy the differential equations
\begin{eqnarray}
\dot{\bar{x}} (t) &=& A  \bar{x} (t)  ,
\label{eq:x-bar}
 \\
  \dot{\Sigma} (t) &=&  A \Sigma (t) + \Sigma (t) A^T + BB^T .
  \label{eq:Sigma}
\end{eqnarray}
These equations characterize the dynamical evolution of the Gaussian distributions of the state variables $x(t)$. Expected values of quadratic forms can easily be evaluated, for instance in the zero mean case  $\mathbb{E} [x^T(t) M x(t) ] = \mathrm{Tr}  [ M^T \Sigma(t)]  $. In the frequency domain, the spectral density $R_{out}[i\omega]$ of the output process $y(t)$ is related to the  spectral density $R_{in}[i\omega]$ of the input process $w(t)$ via the transfer relation
\begin{equation}
R_{out}[i\omega] =  \Xi_G[ i\omega] R_{in}[i\omega]  \Xi_G[i\omega]^\dag  ,
 \label{eq:transfer-intro-1}
\end{equation}
where $\Xi_G$ is the transfer function for the system $G$. Differential equations of the types (\ref{eq:x-bar}) and (\ref{eq:Sigma}), and transfer relations like (\ref{eq:transfer-intro-1}),  play fundamental roles in classical linear systems and signals theory, \cite{AM71,KS72,MHAD77,AM79,KV86}.

{\em Quantum linear systems}  are a class  of open quantum systems fundamental to quantum optics and quantum technology, \cite{BR04,KMN07,Milburn08,WM08,WM10,MNM10}. The equations describing quantum linear  systems (see Section \ref{sec:lqss} below) look formally like the classical equations (\ref{eq:lcss}), but they are not classical equations, and in fact give the Heisenberg dynamics of a system of coupled open quantum harmonic oscillators, \cite{WM10,JNP08}. These quantum systems are driven by quantum fields that describe the influence of the external environment (e.g. light beams) on the oscillators. In quantum optics the fields play the role of {\em quantum signals}. As pointed out in the paper \cite{MNM10}, a large fraction of quantum optics literature concerns fields in coherent states (a type of Gaussian state), and as a consequence quantum optical systems driven by coherent fields are well understood. Indeed, when a quantum linear system, initialized in a Gaussian state, is driven by a Gaussian field, the state of the system is Gaussian, with mean and covariance satisfying equations of the form  (\ref{eq:x-bar}) and (\ref{eq:Sigma}). This fact is of basic importance to Gaussian quantum systems and signals theory, and for example has been exploited for the purpose of quantum $H^\infty$ and LQG control design, \cite{DJ99,WD05,JNP08,NJP09,WM10,ZJ11}.

In recent years, due to their highly non-classical properties,  single-photon light fields have found important applications in quantum communication, quantum computation, quantum cryptography,  and quantum metrology, etc.,  \cite{KMN07,CMR09,MNM10}. Unlike Gaussian states,  a light pulse in single-photon state contains one and only one photon,  and is thus highly non-classical.
While much of the optical  quantum information literature deals with single modes (namely discrete variables like polarization) of light and static devices (like beamsplitters), the importance of continuous mode photons and dynamical devices is becoming clear, \cite{RR05,Milburn08,MNM10}.  However, single photon states  of light are not Gaussian, and so the relatively well developed  quantum Gaussian systems and signals theory  is not directly useful for quantum optical systems driven by single photon fields.

The purpose of this paper is to extend linear systems and signals theory to include single photon quantum signals.
We build on the results in \cite{Milburn08}  to describe how quantum linear dynamical systems respond to multichannel continuous mode photon fields from a system-theoretic point of view. We show, for example, that
when a quantum linear system $G$ with no scattering (equations (\ref{system-a})-(\ref{system-out}) below with $S=I$) is driven by multichannel photon fields, the mean $\bar {\breve a}(t) = \mathrm{Tr} [\rho\breve a(t)]$ and covariance $\Sigma(t) = \mathrm{Tr}[ \rho( \breve a(t) - \bar{\breve{a}}(t) ) (\breve a (t) -  \bar{\breve{a}}(t))^\dag]$, where $\rho$ is the initial joint system-field density operator, satisfy the differential equations
\begin{eqnarray}
\dot{\bar{\breve{a}}} (t) &=& A  \bar{\breve{a}} (t)  ,
\label{eq:x-bar-q}
 \\
  \dot{\Sigma} (t) &=&  A \Sigma (t) + \Sigma (t) A^\dag +  B\Gamma^\dag(\xi(t)) + \Gamma(\xi(t)) B^\dag + BFB^\dag .
  \label{eq:Sigma-q}
\end{eqnarray}
In equation (\ref{eq:Sigma-q}), $F$ is a matrix depending on the Ito products of the input fields \cite{KRP92,GZ00}, and $\Gamma(\xi(t))$ is a matrix depending on the pulse shape matrix $\xi(t)$. Equation (\ref{eq:Sigma-q}) is crucial to the study of intensity of output fields, cf. Sec.~\ref{sec:photon}.

When multichannel photons are input to a quantum linear system, the output state can be quite complex. To accommodate the types of states that can be produced from multichannel photon inputs, we define a class $\mathcal{F}$ of quantum states, which we call {\em photon-Gaussian}  states. A state (density operator) $\rho \in \mathcal{F}$ is specified by a matrix $\xi(t)$ of functions (or pulses, a multichannel generalization of wavepackets), and a Gaussian spectral density  $R[i\omega]$. We  sometimes express these states as $\rho_{\xi, R}$, as in Fig.~\ref{main_figure}. Our main result, Theorem \ref{thm:main}, states that if the photon-Gaussian state  $\rho_{\xi_{in}, R_{in}} \in \mathcal{F}$ is input to a quantum linear system $G$ initialized in the vacuum state, then the steady-state output state is also a photon-Gaussian state  $\rho_{\xi_{out}, R_{out}} \in \mathcal{F}$. Moreover, the transfer $R_{in} \mapsto R_{out}$ is given by the above relation (\ref{eq:transfer-intro-1}), while the transfer $\xi_{in} \mapsto \xi_{out}$ is given by
\begin{eqnarray}
 \xi_{out}[s]  =   \Xi_{G}[s] \xi_{in}[s]  ,
 \label{eq:transfer-intro-2}
\end{eqnarray}
where $\Xi_G$ is the transfer function of the quantum linear system $G$ defined in Sec.~\ref{sec:systems}. This result provides a natural generalization of the well-known Gaussian transfer properties of classical linear systems to an important class of highly non-classical quantum states that includes single photon states. Results of this type are anticipated to be of fundamental importance to the analysis and design of quantum systems for the  processing of   highly non-classical quantum states.
While most of this paper is concerned with questions of analysis, we include a short section on synthesis, generalizing the work \cite{Milburn08}. Here, a quantum linear system is designed to manipulate the wavepacket shape of a single photon. This is an example of {\em coherent control}, \cite{Milburn08}.

\begin{figure}
\centering
\includegraphics[width=2.5in]{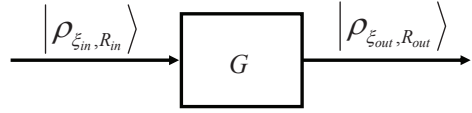}\\
\caption{A quantum linear stochastic system $G$ driven by multichannel field in a photon-Gaussian state (Definition \ref{def:F}). The multichannel input and output field states are denoted $\rho_{\xi_{in},R_{in}}$  and $\rho_{\xi_{out},R_{out}}$  respectively. The system $G$ transfers the input pulse shape $\xi_{in}$  and covariance function $R_{in}$ to the output pulse shape $\xi_{out}$  and covariance function $R_{out}$ respectively. In this paper it is assumed that the system is initially isolated from the input field.}
\label{main_figure}
\end{figure}

{\em Notation}. Given a column vector of complex numbers or   operators $x=[
\begin{array}{ccc}
x_{1} & \cdots & x_{k}%
\end{array}
]^{T}$ where $k$ is a positive integer, define $x^{\#}=[
\begin{array}{ccc}
x_{1}^{\ast} & \cdots & x_{k}^{\ast}%
\end{array}
]^{T}$, where the asterisk $\ast$ indicates complex conjugation or Hilbert space adjoint. Denote $x^\dag = (x^\#)^T$. Furthermore, define the doubled-up column vector to be $\breve {%
x}=[
\begin{array}{cc}
x^{T} & \left( x^{\#}\right) ^{T}%
\end{array}
]^{T}$. The matrix case can be defined analogously. Let $I_{k}$ be an identity matrix and $0_k$ a zero square matrix, both of dimension $k$. Define $J_{k}=\mathrm{diag}%
(I_{k},-I_{k})$ and $\Theta_k =  [\begin{array}{cccc}
0 & I_k; & -I_k & 0%
\end{array}]$ (The
subscript ``$k$'' is omitted when it causes no confusion.) Then for a matrix $X\in\mathbb{C}^{2j\times 2k}$, define $X^{\flat}=J_{k}X^{\dag}J_{j}$. $\otimes_{c}$ denotes the Kronecker product. $m$ is the number of input channels, and $n$ is the number of degrees of freedom of a given quantum linear stochastic system (that is, the number of oscillators). $\vert\phi\rangle$ denotes the initial state of the system which is always assumed to be vacuum, $\vert 0\rangle$ denotes the vacuum state of free fields. Given a function $f(t)$ in the time domain, define its two-sided Laplace transform  \cite[Chapter 10]{WRL61} to be
$F[s] = \mathscr{L}_b \{f(t)\}(s)  := \int_{-\infty}^\infty e^{-st} f(t) dt $.
When $s=i\omega$, we have the Fourier transform
$ F[i\omega] := \int_{-\infty}^\infty e^{-i\omega t} f(t) dt $.
Given two constant matrices $U$, $V\in \mathbb{C}^{r\times k}$, a doubled-up matrix $\Delta\left(U,V\right) $ is defined as
\begin{equation}
\Delta\left(U,V\right):=\left[
\begin{array}{cc}
U & V \\
V^{\#} & U^{\#}%
\end{array}
\right] .
\end{equation}
Similarly, given time-domain matrix functions $E^-(t)$ and $E^+(t)$ of compatible dimensions, define a doubled-up matrix function
\begin{equation}
\Delta(E^-(t),E^+(t)) := \left[ \begin{array}{cc}
                                 E^-(t) & E^+(t) \\
                                 E^+(t)^\# & E^-(t)^\#
                               \end{array}
  \right].
\end{equation}
Then its two-sided Laplace transform is
\begin{equation}
\Delta (E^-[s], E^+[s]) = \mathscr{L}_b \{\Delta(E^-(t),E^+(t))\}(s) = \left[ \begin{array}{cc}
E^-[s] & E^+[s] \\
E^+[s^\ast]^\# & E^-[s^\ast]^\#
\end{array}\right] .
\end{equation}
Finally, given two operators $A$ and $B$, their commutator is defined to be $[A,B]=AB-BA$.

\section{Quantum Linear Systems} \label{sec:systems}

\begin{figure}
\centering
\includegraphics[width=2.0in]{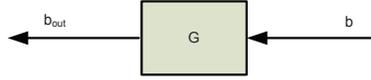}
\caption{A quantum linear stochastic system $G$ with input $b$ and output $b_{out}$}
\label{system}
\end{figure}

In this section quantum signals and systems of interest are introduced, Fig.~\ref{system}. Quantum systems behave in accordance with the laws of quantum mechanics, and in Sec.~\ref{sec:lqss} we summarize the  dynamics  of a quantum linear system driven by external quantum fields. These models feature inputs and outputs, corresponding, for example, to light incident on and reflected by the system, \cite{GZ00,WM10}. In Sec.~\ref{sec:i_o} we write down explicitly the input-output relations, using a notation for the impulse response motivated by physical annihilation and creation processes. Since one of our objectives is to study the steady-state response of the quantum linear system, we present in Sec.~\ref{sec:i_o_stationary} the steady-state versions of the input-output relations, as well as the transfer function defined in terms of the two-sided Laplace transform. Our analysis will require the stable inversion of the transfer function, and this is presented in Sec.~\ref{sec:identity}. Sections \ref{sec:states-mode}, \ref{sec:states-gaussian-field} and \ref{sec:states-photon-field} provide the definitions of Gaussian and photon states needed in this paper. Sec.~\ref{sec:covariance-transfer} then describes how output means and covariances are determined from the corresponding input quantities.

\subsection{Dynamics}\label{sec:lqss}

The open quantum linear system $G$, shown in Fig.~\ref{system}, is a collection of $n $ interacting quantum harmonic oscillators $a=[a_{1},\ldots
,a_{n}]^{T}$ (defined on a Hilbert space $\mathfrak{H}_G$)  coupled to $m$ boson fields $b(t)=[b_1(t), \ldots, b_m(t)]^{T}$ (defined on a Fock space $\mathfrak{F}$)
\cite{WM08,WM10,ZJ11}. Here, $a_{j}$ ($j=1,\ldots,n$) is the annihilation operator of the $j$th quantum harmonic oscillator satisfying the canonical commutation relations $[a_{j},a_{k}^{\ast}]=\delta_{jk}$. The vector $b(t)$ represents an $m$-channel electromagnetic field in free space, which satisfies the singular commutation relations
\begin{equation} \label{eq:ccr}
\lbrack b_j(t),b_k^{\ast }(t^{\prime })]=\delta _{jk}\delta
(t-t^{\prime }),~ [b_j(t), b_k(t^{\prime })]=[b_j^{\ast
}(t),b_k^{\ast }(t^{\prime })]=0, ~ j,k=1,\ldots, m, ~ \forall t,t^{\prime} \in \mathbb{R} .
\end{equation}%
The operator $b_j(t)$ ($j=1,\ldots,m$) may be regarded as a quantum stochastic process;
in the case where the field is in the vacuum state (denoted $\vert 0 \rangle$, \cite{KRP92,GZ00,WM10}), this process is quantum white noise. The integrated field operators are given by  $B(t)=\int_{t_0}^t b(r) dr$ and $B^\# (t)=\int_{t_0}^t b^\#(r) dr$, which are quantum Wiener processes.
The gauge process is given by
$\Lambda(t) = \int_{t_0}^t b^\#(\tau)b^T(\tau)d\tau = \left(\Lambda_{jk}(t)\right)_{j,k=1,\ldots,m}$
with operator entries $\Lambda_{jk}$ on the Fock space $\mathcal{F}$. Finally in this paper it is assumed that these quantum stochastic processes are \emph{canonical}, that is, they have the following non-zero Ito products
\begin{align}
 dB_j(t) dB_k^\ast(t) =& ~ \delta_{jk} dt, ~ d\Lambda_{jk}dB_l^\ast(t) = \delta_{kl}dB_j^\ast(t),  \label{Eq:CCR2}\\
 dB_j(t)d\Lambda_{kl}(t) =& ~ \delta_{jk}dB_l(t), ~ d\Lambda_{jk}(t)d\Lambda_{lm}(t) = \delta_{kl}d\Lambda_{jm}(t), ~ (j,k,l=1, \ldots, m) . \nonumber
\end{align}

The system $G$ can be parameterized by a triple $\left(
S_{-},L,H\right) $. In this triple, $S_{-}$ is a scattering matrix
(satisfying $S_{-}S_{-}^{\dag}=S_{-}^{\dag}S_{-}=I_{m}$). The vector operator $L$ is
defined as
$L=C_{-}a+C_{+}a^{\#}$,
where $C_{-}$ and $C_{+}\in\mathbb{C}^{m\times n}$. The Hamiltonian
$H=\frac{1}{2}\breve{a}^\dag\Delta\left(\Omega_{-},\Omega_{+}\right) \breve{a}$
describes  the initial internal energy of the oscillators,
where $\Omega_{-}, \Omega_{+}\in C^{n\times n}$ satisfy $\Omega
_{-}=\Omega_{-}^{\dag}$ and $\Omega_{+}=\Omega_{+}^{T}$.  With these parameters, the Schrodinger's equation for the system $G$ (with initial internal energy $H$) and boson field is, in Ito form (\cite[Chapter 11]{GZ00}),
\begin{equation} \label{eq:U}
dU(t,t_0) = \left\{ \mathrm{Tr}[(S_- - I_m) d\Lambda^T] + dB^\dag(t)L - L^\dag S_- dB(t) - (\frac{1}{2} L^\dag L + iH )dt \right\} U(t,t_0),
\ \ t \geq t_0,
\end{equation}
with $U(t_0,t_0)=I$ (the identity operator) for all $t \leq t_0$. Moreover, by means of Ito rules, it can be shown that the operator $U(t,t_0)$ satisfies  $d(U(t,t_0)U^\ast(t,t_0)) = d(U^\ast(t,t_0)U(t,t_0))= 0, \ \forall t\geq t_0$, and since $U(t_0,t_0)=I$ we see that $U(t,t_0)$ is unitary.

In the Heisenberg picture, the system operators evolve according to $\breve{a}(t) = U(t,t_0)^\ast \breve{a} U(t,t_0)$ (componentwise on the components of $\breve a$). The output field $\breve b_{out}(t)$ (which carries away information about the system after interaction) is defined by $\breve{b}_{out}(t) = U(t,t_0)^\ast \breve{b}(t) U(t,t_0)$ (componentwise on the components of $\breve{b}(t)$). Consequently, by Eq. (\ref{eq:U}), the dynamical model for the system $G$  can be written as
\begin{align}
 \dot{\breve{a}}(t) & =A\breve{a}(t) +BS\breve{b}\left( t\right) , \ \ \breve{a}(t_0) = \breve{a},   \label{system-a} \\
 \breve{b}_{out}\left( t\right) & =C\breve{a}(t)+S\breve{b}\left(
t\right) ,  \label{system-out}
\end{align}
in which system matrices are given in terms of  physical parameters by
$S=\Delta\left( S_{-},0\right), C=\Delta\left( C_{-},C_{+}\right)
, B=-C^{\flat}, A=-\frac{1}{2}C^{\flat}C-iJ_{n}\Delta\left(\Omega_{-},\Omega_{+}\right) $.
\begin{remark}{\rm Eqs. (\ref{system-a}) and (\ref{system-out}) are quantum linear systems, which can be obtained using the definitions $\breve{a}(t) = U^\ast(t,t_0)\breve{a}(t_0)U(t,t_0)$ and $\breve{b}_{out}(t) = U^\ast(t,t_0)\breve{b}(t)U(t,t_0)$ and making use of Ito rules and the  commutation relations for annihilation and creation operators, \cite{WM10}.}
\end{remark}

The system $G$ is said to be {\it asymptotically stable} (equivalently, \emph{exponentially stable}) if the matrix $A$ is Hurwitz \cite[Sec.~III-A]{ZJ11}.  More details on quantum linear stochastic systems can be found in, e.g., \cite{GZ00}, \cite{WM10}, \cite{ZJ11}, \cite{ZJ12} and references therein.

\begin{remark}{\rm
Quantum linear systems have been widely used in quantum optics. Moreover, They have also been used in opto-mechanical systems, e.g., Eqs. (15)-(18) in the Supplementary Information of Reference \cite{MHP+11} and Eqs. (9)-(10) and the line below Eq. (10) in Reference \cite{T12}.  They also appear in circuit quantum electrodynamics (circuit QED) systems, e.g., Eqs. (18)-(21) in Reference \cite{MJP+11}.
}
\end{remark}

\subsection{Input-output relations}\label{sec:i_o}

The output may be expressed in terms of the input and initial system variables
\begin{equation}\label{eq:out_tf_0}
\breve{b}_{out}(t) = Ce^{A(t-t_0)}\breve{a} +  \int_{t_0}^{t} C e^{A(t-r)} BS
\breve{b}(r)dr + S \breve{b}(t),  \ \ t\geq t_0.
\end{equation}
We find it convenient to express this input-output relation in terms of impulse response functions. Define
\begin{eqnarray}
g_{G^{-}}(t)&:=&\left\{
\begin{array}{ll}
\delta(t)S_{-}-[%
\begin{array}{cc}
C_{-} & C_{+}%
\end{array}
]e^{At}\left[
\begin{array}{c}
C_{-}^{\dag} \\
-C_{+}^{\dag}%
\end{array}
\right] S_{-}, & t\geq 0   \\
0, & t<0 %
\end{array}
\right., \nonumber \\
\ g_{G^{+}}(t)&:=&\left\{
\begin{array}{ll}
-[%
\begin{array}{cc}
C_{-} & C_{+}%
\end{array}
]e^{At}\left[
\begin{array}{c}
-C_{+}^{T} \\
C_{-}^{T}%
\end{array}
\right]S_-^\# , & t\geq 0  \\
0 , & t<0  %
\end{array}
\right. .  \label{eq:io}
\end{eqnarray}
(Later we use $g_{G^-}^{jk}$ and $g_{G^+}^{jk}$ ($j,k=1,\ldots,m$) to denote the entries of $g_{G^-}$ and $g_{G^+}$ on the $j$th row and $k$th column, respectively.) The \emph{impulse response function} for the system $G$ is
\begin{equation}\label{eq:gg}
g_{G}(t):=\left\{
\begin{array}{ll}
\delta(t)S-Ce^{At}C^{\flat}S, & t\geq 0   \\
0, & t<0 %
\end{array}
\right. .
\end{equation}
It can be checked  that $g_{G}(t)$ defined in Eq. (\ref{eq:gg}) is in the form of
\begin{equation}
g_{G}(t)=\Delta\left( g_{G^{-}}(t),g_{G^{+}}(t)\right) .  \label{eq:impulse}
\end{equation}
Therefore the input-output relation (\ref{eq:out_tf_0}) may be expressed in a more compact form
\begin{equation}\label{eq:out_tf}
\breve{b}_{out}(t) = Ce^{A(t-t_0)}\breve{a} +  \int_{t_0}^{t}g_{G}(t-r)\breve{b}(r)dr.
\end{equation}

\subsection{Steady-state input-output relations}\label{sec:i_o_stationary}

Assume that the system (\ref{system-a}) is asymptotically stable.  Letting $t_0 \to  -\infty$ and noticing (\ref{eq:gg}), Eq. (\ref{eq:out_tf}) becomes
\begin{equation} \label{eq:out_tf3}
\breve{b}_{out}(t) = \int_{-\infty}^\infty g_{G}(t-r)\breve{b}(r)dr.
\end{equation}

Let $\Xi_G[s]$, $\Xi_{G^-}[s]$, $\Xi_{G^+}[s]$, $\breve b[s]$ and $\breve{b}_{out}[s]$ denote the two-sided Laplace transforms of $g_G(t)$, $g_{G^-}(t)$, $g_{G^+}(t)$, $\breve b(t)$ and $\breve b_{out}(t)$, respectively. Then by the above definitions and standard properties of the two-sided Laplace transform we have the transfer function relation
$\breve{b}_{out}[s]  =\Xi_G[s]\breve{b}[s]$, where
\begin{equation}\label{eq:G_omega}
\Xi_G[s] = \Delta(  \Xi_{G^-}[s], \Xi_{G^+}[s] ) .
\end{equation}

\subsection{Flat-unitary property and stable inversion}\label{sec:identity}

For later use we record here some important inversion results for quantum linear systems.

\begin{proposition} \label{prop:identity}   
The reciprocal of the transfer function  $\Xi_G[s]$ is given by
\begin{equation}
\Xi_G[s]^{-1} = \Xi_G[-s^\ast]^\flat = \Delta(  \Xi_{G^-}[-s^\ast]^\dag, - \Xi_{G^+}[-s]^T ).
\label{eq:G_flat}
\end{equation}%
Consequently, we have the fundamental flat-unitary relation
\begin{equation}\label{eq:unitary}
\Xi_G[i\omega]^\flat  \Xi_G[i\omega] = \Xi_G[i\omega] \Xi_G[i\omega]^\flat = I_{2m},~~ \forall \omega \in \mathbb{R} .
\end{equation}
\end{proposition}

\begin{proof}
Equation (\ref{eq:G_flat}) follows from \cite[Eq. (74)]{GJN10a}, and the flat-unitary relation (\ref{eq:unitary}) follows on setting $s=i\omega$.
\end{proof}

We now apply the results and methods of \cite{Sog10} to find a stable inverse of the quantum linear system $G$.
Define
\begin{equation}
g_{G^{-1}}(t):=\mathscr{L}_{b}^{-1}\{   \Xi_G[s]^{-1}  \}(t) ,
\label{eq:stable-inverse}
\end{equation}
where $\mathscr{L}_{b}^{-1}$ is the inverse two-sided Laplace transform \cite[Chapter 10]{WRL61}.

On the basis of Proposition \ref{prop:identity}, the following result can be established.
\begin{lemma}\label{lem:G_inv}
Assume that the system  $G$ (Eq. (\ref{system-a})) is asymptotically stable. Then the impulse response of the stable inverse of $G$ is given by
\begin{equation}
g_{G^{-1}}(t)=\Delta \left( g_{G^{-}}(-t)^{\dag },-g_{G+}(-t)^{T}\right) = \left\{
\begin{array}{ll}
0, & t > 0 \\
S^{\flat}\delta(t)-S^{\flat}Ce^{-A^{\flat}t}C^{\flat}, & t \leq 0%
\end{array}
\right. .
\label{eq:G-inv}
\end{equation}
\end{lemma}
\begin{proof}
By definition of the two-sided Laplace transform, we have $\mathscr{L}_b\{  f(-t)^\ast \} (s) = f[-s^\ast]^\ast$ and $\mathscr{L}_b\{  f(-t) \} (s) = f[-s]$. Combining this with expression (\ref{eq:G_flat}) in Proposition \ref{prop:identity}, we obtain the first equality in (\ref{eq:G-inv}). The second equality in (\ref{eq:G-inv}) follows from the definitions of the impulse responses $g_{G^-}(t)$ and $g_{G^+}(t)$ in terms of the system matrices.
\end{proof}

\subsection{Gaussian system (single mode) states}
\label{sec:states-mode}

A stationary Gaussian system state $\rho_s $ on $ \mathfrak{H}_G$ is specified by its characteristic function \cite{Gough05},
\begin{equation}\label{eq:system-gaussian_multi}
\mathrm{Tr}[\rho_s\exp(i\breve{z}^\dag\breve{a})] = \exp(-\frac{1}{2}\breve{z}^\dag \Sigma \breve{z}
+i \breve z^\dag \breve \beta
),  ~~ \forall z \in \mathbb{C}^n ,
\end{equation}
where $\breve{\beta} = \bar {\breve a}(t) = \mathrm{Tr}[\rho_s \breve a(t)]$, and
 $\Sigma =  \mathrm{Tr}[\rho_s ( \breve a (t)- \breve{\beta}) (\breve a (t) -  \breve{\beta}   )^\dag]$  is a non-negative Hermitian matrix.
 In general, $\Sigma$ has the form
 \begin{equation}
\Sigma = \left[ \begin{array}{cc}
             I_n +N^T  & M \\
             M^\dag & N
           \end{array}
  \right] .
\end{equation}
  In particular, the \emph{ground} or\emph{ vacuum} state $\vert\phi\rangle$ is specified by $\beta =0$ and $  \Sigma_\phi = \left[ \begin{array}{cc}
             I_n & 0 \\
             0 & 0_n
           \end{array}
  \right]$.

\subsection{Gaussian field states} \label{sec:states-gaussian-field}

Depending on the nature of the boson fields, input signals $b(t)$ can be in various states. In this section we consider $m$-channel Gaussian field states $\rho_g$, \cite{Gough05}. Given a function $f \in L^2(\mathbb{R},\mathbb{C}^m)$, define an integral functional $\mathfrak{R}$ on the Hilbert space $L^2(\mathbb{R},\mathbb{C}^{2m})$ to be
\begin{equation}
(\mathfrak{R}\breve{f})(t) : = \int_{-\infty}^\infty R_g(t,r)\breve{f}(r)dr ,
\end{equation}
where $R_g(t,r)=R_g(r,t)^\dag \geq 0$. The $m$-channel Gaussian field state can be specified by the characteristic function
\begin{equation}
\mathrm{Tr}[\rho_g\exp(i\int_{-\infty}^\infty \breve{f}(t)^\dag \breve{b}(t)dt)] = \exp(-\frac{1}{2}\langle\breve{f}\vert \mathfrak{R} \breve{f}\rangle + i \langle \mathfrak{\breve{m}}\vert \breve{f} \rangle),
\end{equation}
where $ \mathfrak{m} \in L^2(\mathbb{R},\mathbb{C}^m)$, and $\langle \cdot \vert \cdot \rangle$ denotes the inner product in the Hilbert space $L^2(\mathbb{R},\mathbb{C}^{2m})$. It can be checked that the mean value is $\mathrm{Tr}[\rho_g\breve{b}(t)] = \mathfrak{\breve{m}}(t)$ and the covariance function is
\begin{equation}
\mathrm{Tr}[\rho_g(\breve{b}(t)-\mathfrak{\breve{m}}(t))(\breve{b}^\dag(r)-\mathfrak{\breve{m}}(r)^\dag)] = R_g(t,r), ~~ \forall r,t \in \mathbb{R}.
\end{equation}
In general, the covariance function $R_g(t,r)$ has the form
\begin{equation}\label{R_field}
R_g(t,r) = \left[
\begin{array}{cc}
  \delta(t-r)+\nu(r,t)^T & \mu(t,r) \\
  \mu(r,t)^\dag & \nu(t,r)
\end{array}
\right]
\end{equation}
for real $r$ and $t$. When $\rho_g$ is stationary, $R_g(t,r)$ depends on the difference between $t$ and $r$, instead of their particular vales. In this case, we may use $R_g(\tau)$ to replace  $R_g(t,r)$. In particular, the quantum \emph{vacuum} field state $\vert 0\rangle$ is specified by
\begin{equation}\label{R_in_vacuum}
R_0 (\tau) =  \delta(\tau)\left[
\begin{array}{cc}
  I_m & 0 \\
  0 & 0_m
\end{array}
\right].
\end{equation}

\subsection{Photon field states} \label{sec:states-photon-field}

Now we introduce another type of field states:  the continuous-mode single photon pure states. In the one channel case we denote it by $\vert \Psi \rangle = \vert 1_\nu \rangle$, as defined in \cite[Eq.~(6.3.4)]{RL00}, \cite[Eq.~(9)]{Milburn08}
\begin{equation}
\vert 1_\nu \rangle = B^\ast(\nu) \vert 0  \rangle,
\label{eq:xi-create}
\end{equation}
where $\vert 0 \rangle$ is the vacuum state of the field as defined in Sec.~\ref{sec:states-gaussian-field}, and $B^\ast(\nu) := \int_{-\infty}^\infty \nu(r) b^\ast(r)dr$. Here, $\nu$ is a complex-valued function such that $\int_{-\infty}^\infty \vert \nu(r) \vert^2 dr = 1$.
Expression (\ref{eq:xi-create}) says that  the single photon wavepacket is created from the vacuum using the field operator $B^\ast(\nu)$. Using the relation $b(t) \vert 1_\nu \rangle = \nu(t) \vert 0 \rangle$,  we see that the operator $B(\nu) := \int_{-\infty}^\infty \nu(r)^\ast b(r)dr $  annihilates a photon, resulting in  the vacuum: $B(\nu) \vert 1_\nu \rangle  = \vert 0 \rangle$. 

The field operator $b(t)$ is a zero mean quantum stochastic process with respect to the single photon state.
The covariance function is given by
\begin{eqnarray}
R(t,r) =
\mathbb{E}_{1_\nu}[ \breve b(t) \breve b^\dag(r) ]
=\delta(t-r) \left[\begin{array}{cc}
                     1 & 0 \\
                     0 & 0
                   \end{array}
 \right] +
\left[ \begin{array}{cc}
 \nu(r)^\ast \nu(t)  &  0
\\
0 & \nu(r)\nu(t)^\ast
\end{array} \right].
\label{eq:R-photon}
\end{eqnarray}

The gauge process $\Lambda(t)$ (recall Section \ref{sec:lqss}) for a single photon channel takes the form $\Lambda(t) = \int_0^t n(r) dr$, where $n(t) = b^\ast(t) b(t)$ is the number operator for the field. The intensity of the field is the mean $\bar n(t) = \langle 1_\nu \vert n(t) \vert 1_\nu \rangle = \vert \nu(t) \vert^2$, an important physical quantity that determines the probability of photodetection per unit time.

\subsection{Mean and covariance transfer}
\label{sec:covariance-transfer}  

In this section  some basic covariance transfer results for the quantum linear system $G$ are presented, which are the quantum counterparts of the well-known classical results, e.g. \cite[Sec. 1.10.3]{KS72}, but adapted to take into account the non-commuting system and field variables.

Consider a quantum linear system $G$ initialized at time $t_0$ in a state $\rho = \vert\phi\rangle\langle\phi\vert \otimes \rho_{field}$, where $\vert \phi\rangle$ is the vacuum  system state and the field state $\rho_{field}$ satisfies $\mathrm{Tr}[\rho_{field}\breve{b}(t)]=0$ (particular choices of $\rho_{field}$ will be made below).
By taking expectations in Eq. (\ref{system-a}), we find that the mean $\bar {\breve a}(t) =\mathrm{Tr}[\rho\breve a(t)]$ satisfies Eq. (\ref{eq:x-bar-q}) since the field has mean zero. Also, since the system is initialized in the ground state $|\phi\rangle$, we have $\bar {\breve a}(t) = \bar {\breve a}(t_0) =0$ for all $t \geq t_0$.
Define the matrix
\begin{equation} \label{eq:Gamma}
\Gamma(t) = \mathrm{Tr}[\rho \breve a(t) \breve b^\dag (t)],
\end{equation}
and a matrix $F$ by
\begin{equation} \label{eq:F}
F dt  =\mathrm{Tr}[\rho_{field} d\breve B(t) d \breve B^\dag(t)].
\end{equation}

\begin{lemma}   \label{lemma:Sigma-general}
The quantum linear system $G$ initialized in the state $\rho$ (ground system state and zero mean field state) has zero mean $\bar {\breve a}(t) =0$ for all $t \geq t_0$ and the covariance matrix
$\Sigma(t) = \mathrm{Tr}[\rho\breve a(t) \breve a^\dag(t)]$
 satisfies the differential equation
 \begin{eqnarray}
  \dot{\Sigma} (t) &=&  A \Sigma (t) + \Sigma (t) A^\dag   + BS \Gamma^\dag(t) + \Gamma(t) S^\dag B^\dag + BSFS^\dag B^\dag,
  \label{eq:Sigma-dyn-general}
\end{eqnarray}
with initial condition $\Sigma (t_0) = \Sigma_\phi$, where $F$ is given in Eq. (\ref{eq:F}).
\end{lemma}

\begin{proof}
The proof of Eq. (\ref{eq:Sigma-dyn-general}) follows by taking expectations of the  differential
\begin{align}
d( \breve a(t) \breve a^\dag(t) ) &=  (d \breve a(t) ) \breve a^\dag(t) + \breve a(t) d \breve a^\dag(t) + d \breve a(t) d \breve a^\dag(t) \\
&=
A  \breve a(t) \breve a^\dag(t)dt  + BS d\breve B(t) \breve a^\dag(t) +  \breve a(t) \breve a^\dag(t)  A^\dag dt + \breve a(t) d\breve B^\dag(t) S^\dag B^\dag + BS F S^\dag B^\dag dt  \nonumber
\end{align}
and noticing that $\mathrm{Tr}[\rho d\breve{B}(t)\breve{a}^\dag(t)]= \mathrm{Tr}[\rho\breve{a}(t)d\breve{B}^\dag(t)]^\dag$.
\end{proof}

Note that expected values of quadratic forms may easily be evaluated in terms of $\Sigma(t)$, for example, $\mathrm{Tr}[\rho \breve a^\dag (t) M \breve a(t)]
= \mathrm{Tr}[M^T \Sigma(t)^T] - \mathrm{Tr}[M^T J]$.

\begin{remark}{\rm
Further information regarding the dynamics of $\Gamma(t)$ will be given in Section \ref{sec:photon-t0} for the case of multichannel photon fields.
}
\end{remark}

Suppose now $\breve b_{out}(t)$ is the steady-state output field defined by (\ref{eq:out_tf3}).
Define the input and output covariances $R_{in}(t,r):=\mathrm{Tr}[\rho\breve b(t) \breve b^\dag (r)]$  and $R_{out}(t,r):= \mathrm{Tr}[\rho \breve b_{out}(t) \breve b_{out}^\dag (r)]$.

\begin{theorem}\label{thm:covariance-transfer}
Assume that the system $(\ref{system-a})$ is asymptotically stable. Let the input field have covariance $R_{in}(t,r)$. Then the steady-state output covariance   is given by
\begin{eqnarray}
R_{out}(t,r) = \int_{-\infty}^\infty \int_{-\infty}^\infty g_{G}(t-\tau_1) R_{in}(\tau_1, \tau_2) g_G(r-\tau_2)^\dag d\tau_1 d\tau_2 .
\label{eq:covariance-transfer}
\end{eqnarray}
\end{theorem}

Now suppose that $\breve b(t)$ is stationary with respect to the field state $\rho_{field}$. Write $R[i\omega]$ for the spectral density matrix (where $R[s]$ is the two-sided Laplace transform of $R(\tau)$).

\begin{theorem}     \label{thm:spectral-transfer}
Assume that the system $(\ref{system-a})$ is asymptotically stable. Let the input field have spectral density matrix  $R_{in}[i\omega]$. Then the output spectral density matrix is given by
\begin{eqnarray}
R_{out}[i\omega] = \Xi_G[ i\omega] R_{in}[i\omega]  \Xi_G[i\omega]^\dag  .
\label{eq:R-transfer}
\end{eqnarray}
\end{theorem}

If the input field state is vacuum $\rho_{in} = \vert0\rangle\langle0\vert$, then the steady-state output field state $\rho_{out}$ is a Gaussian state with covariance $R_{out}(\tau)$ which is in general not the vacuum state. However, if $C_+=0$ and $\Omega_+=0$ (\emph{passive} systems) then we have
\begin{equation} \label{R_out_vac}
R_{out} [i\omega]  = \Delta ( \Xi_{G^-}[ i\omega], 0)
\left[  \begin{array}{cc}
I & 0
\\
0 & 0
\end{array} \right]  \Delta ( \Xi_{G^-}[ i\omega], 0)^\dag
=\left[  \begin{array}{cc}
I & 0
\\
0 & 0
\end{array} \right] ,
\end{equation}
since for a passive system $\Xi_{G^-}[ i\omega] \Xi_{G^-}[ i\omega]^\dag =I$ (Eq. (\ref{eq:unitary})).
So in the passive case the output state is again the vacuum state.

\section{Output intensities of quantum linear systems driven by multichannel photon inputs}
\label{sec:photon}

We begin our study of the response of quantum linear systems to photon inputs by determining the statistics of the output field. Specifically, we consider multi-channel input signals,  each channel with one photon, as defined in Section \ref{sec:photon-def}, and then we find expressions for the  output intensities (transient, Section \ref{sec:photon-t0}, and steady state, Section \ref{sec:photon-steady}) and correlation, Section \ref{sec:photon-steady-covariance}.

\subsection{Multichannel photon fields}
\label{sec:photon-def}

We now consider $m$ field channels each of which is in a single photon state $\vert 1_{\nu_k} \rangle$, determined by possibly distinct pulse shapes $\nu_k$ satisfying the normalization condition $\int_{-\infty}^\infty \vert\nu_k(t)\vert^2 dt = 1$, $k=1,\ldots,m$.
This means that the state of the $m$-channel input is given by the tensor product
\begin{equation}\label{eq:multichannel}
\vert \Psi_{\nu} \rangle  = \vert 1_{\nu_1} \rangle \otimes \cdots \otimes  \vert 1_{\nu_m} \rangle
= \prod_{k=1}^m B^\ast_k( \nu_k) \vert 0^{\otimes m} \rangle ,
\end{equation}
where the $m$-channel  vacuum state is denoted $\vert 0^{\otimes m} \rangle = \vert 0 \rangle \otimes \cdots \otimes \vert 0 \rangle$. Here, $B^\ast_k(\nu_k)=\int_{-\infty}^\infty \nu_k(t) b^\ast_k(t) dt$ is the creation operator for the $k$-th field channel.

For convenience, we let $\mathcal{F}_0$ denote the class of $m$-channel photon input field states:

\begin{equation}\label{class_F0}
\mathcal{F}_0= \left\{\vert\Psi_{\nu}\rangle = \prod\limits_{k=1}^{m}B_k^\ast (\nu_k)\vert 0^{\otimes m} \rangle : \int_{-\infty}^\infty \vert\nu_k(t)\vert^2 dt = 1, \  k=1,2,\ldots, m  \right\}.
\end{equation}

\subsection{Output intensity when the system is initialized at time $t_0$} \label{sec:photon-t0}

In this section, we study the  intensity of output fields of the system $G$ driven by the $m$-channel input field in the class $\mathcal{F}_0$  (Eq. (\ref{class_F0})). We define this     intensity to be
\begin{equation}\label{eq:nout6}
\bar{n}_{out}(t):= \left\langle \phi \Psi_{\nu} \vert b_{out}^{\#}(t)b_{out}^{T}(t)\vert \phi \Psi_{\nu} \right\rangle_{m \times m} .
\end{equation}

We first introduce some notations.
For each $k=1,\ldots, m$, define
\begin{equation}
\vert\zeta_{k}\rangle:= \prod_{j\neq k} B_j^\ast(\nu_j)\vert 0^{\otimes m}\rangle =\vert 1_{\nu_1}\rangle\otimes \cdots \otimes  \vert 1_{\nu_{k-1}}\rangle\otimes \vert 0\rangle\otimes \vert 1_{\nu_{k+1}}\rangle\otimes \cdots \otimes  \vert 1_{\nu_m}\rangle .  \label{eq:zeta_k}
\end{equation}%
Denote
\begin{equation} \label{xi_in_photon}
\xi_{in}^-(t) = \mathrm{diag}( \nu_1(t), \ldots, \nu_m(t) ) ,
\end{equation}
and
\begin{equation}
C^{-}=[I_m \  0_m ]C, \ B^{-}=BS\left[\begin{array}{c}
                            I_m \\
                            0_m
                          \end{array}\right], \  B^{+}=BS\left[\begin{array}{c}
                            0_m \\
                            I_m
                          \end{array}\right].
\end{equation}
For $t \geq t_0$,  define classical variables%
\begin{equation}\label{eq:m-0}
m_-(t) := \left[
\begin{array}{c}
  \langle \phi\Psi_\nu \vert \breve{a}^\dag(t) \vert \phi\zeta_1 \rangle \\
  \vdots \\
  \langle \phi\Psi_\nu \vert \breve{a}^\dag(t) \vert \phi\zeta_m \rangle
\end{array}
\right]_{m \times 2n} ,  ~~ m_+(t) := \left[
\begin{array}{c}
 \langle \phi\zeta_1  \vert \breve{a}^\dag(t) \vert \phi\Psi_\nu  \rangle \\
  \vdots \\
  \langle \phi\zeta_m  \vert \breve{a}^\dag(t) \vert \phi\Psi_\nu  \rangle
\end{array}
\right]_{m \times 2n} ,
\end{equation}
and%
\begin{equation}
\Sigma_\nu(t) := \langle \phi\Psi_\nu \vert \breve{a}(t)\breve{a}^\dag(t) \vert \phi\Psi_\nu \rangle_{2n\times 2n}, ~ t \geq t_0 .
\label{eq:mm}
\end{equation}
We have the following result.
\begin{theorem}\label{thm:n_out}
When the system $G$ is driven by an $m$-channel photon input field $\vert\Psi_{\nu}\rangle$ in the class $\mathcal{F}_0$ (Eq. (\ref{class_F0})), the   output intensity $\bar{n}_{out}(t)$ is given by
\begin{align}
\bar{n}_{out}(t) = & C^{-\#}\Sigma_\nu(t)^T C^{-T}-C^{-\#}J_n C^{-T}+S_{-}^{\#}\xi_{in}^-(t)^\dag\xi_{in}^-(t)S_{-}^{T}  \label{eq:n} \\
& +S_{-}^{\#}\xi_{in}^-(t)^\dag m_{-}(t)^\#C^{-T}+C^{-\#}m_{-}(t)^T\xi_{in}^-(t)S_{-}^{T},  \nonumber
\end{align}%
where $m_{-}(t)$, $m_{+}(t)$, and $\Sigma_\nu(t)$, defined in Eqs. (\ref{eq:m-0})-(%
\ref{eq:mm}), satisfy the following differential equations:
\begin{equation}
\dot{m}_-(t) = m_-(t)A^\dag +\xi_{in}^-(t)^\dag B^{-\dag},   ~ t \geq t_0 \label{eq:m-}
\end{equation}%
\begin{equation}
\dot{m}_+(t) = m_+(t)A^\dag +\xi_{in}^-(t) B^{+\dag},  ~ t \geq t_0,  \label{eq:m+}
\end{equation}%
and
\begin{align}
\dot{\Sigma}_\nu(t) =&  A\Sigma_\nu(t)+\Sigma_\nu(t)A^\dag+B^-\xi_{in}^-(t)m_-(t)+B^+\xi_{in}^-(t)^\dag m_+(t)+ m_-(t)^\dag \xi_{in}^-(t)^\dag B^{-\dag} \label{eq:m}\\
& +m_+(t)^\dag \xi_{in}^-(t) B^{+\dag}+BS\mathrm{diag}(I_m,0_m)S^\dag B^\dag, ~ t \geq t_0, \nonumber
\end{align}
respectively, with initial conditions $m_-(t_0) = 0$, $m_+(t_0) = 0$, and $\Sigma_\nu(t_0) = \mathrm{diag}(I_m,0_m)$.
\end{theorem}
\begin{proof}
It is straightforward to derive Eqs. (\ref{eq:m-}) and (\ref{eq:m+}). To apply Lemma \ref{lemma:Sigma-general} to establish Eq. (\ref{eq:m}), it suffices to evaluate $\Gamma(t)$ and $F$ defined in Eqs. (\ref{eq:Gamma}) and (\ref{eq:F}) respectively. Notice that
\begin{equation}
\langle \phi\Psi_\nu \vert \breve{a}(t) b^\dag(t) \vert \phi\Psi_\nu \rangle
= m_-(t)^\dag \xi_{in}^-(t)^\dag , ~~ \langle \phi\Psi_\nu \vert \breve{a}(t) b^T(t) \vert \phi\Psi_\nu \rangle
= m_+(t)^\dag \xi_{in}^-(t).
\end{equation}
We have
\begin{equation}\label{eq:Gamma2}
 \Gamma(t) = \left[ \begin{array}{cc}
                      m_-(t)^\dag \xi_{in}^-(t)^\dag & m_+(t)^\dag \xi_{in}^-(t)
                    \end{array}
  \right] .
\end{equation}
On the other hand, for the single-photon input field, $F$ defined in Eq. (\ref{eq:F}) is
\begin{equation}\label{eq:F2}
F dt = \langle\Psi_\nu\vert d\breve B(t) d \breve B^\dag(t)\vert\Psi_\nu\rangle = \left[\begin{array}{cc}
                                                     I_m & 0 \\
                                                     0 & 0
                                                   \end{array}
 \right]dt.
\end{equation}
Substitution of Eqs. (\ref{eq:Gamma2})-(\ref{eq:F2}) into Eq. (\ref{eq:Sigma-dyn-general}) yields Eq. (\ref{eq:m}).
Finally, note that the gauge process of the output field $\Lambda _{out}(t)$
satisfies, e.g. \cite[Sec. IV]{GJ09},
\begin{align}
d\Lambda _{out}(t)& =S_{-}^{\#}d\Lambda
(t)S_{-}^{T}+S_{-}^{\#}dB^{\#}(t)L^{T}(t)+L^{\#}(t)dB^{T}(t)S_{-}^{T}+L^{%
\#}(t)L^{T}(t)dt \\
& =S_{-}^{\#}d\Lambda (t)S_{-}^{T}+S_{-}^{\#}dB^{\#}(t)\breve{a}%
^{T}(t)C^{-T}+C^{-\#}\breve{a}^{\#}(t)dB^{T}(t)S_{-}^{T}+C^{-\#}%
\breve{a}^{\#}(t)\breve{a}^{T}(t)C^{-T}dt . \nonumber
\end{align}%
Noticing that $d \Lambda_{out}  = b_{out}^{\#}(t)b_{out}^{T}(t) dt$ and $\breve{a}(t)\breve{a}^\dag(t) = ( \breve{a}(t)^\#\breve{a}^T(t) )^T + J_n $, it can be readily shown that $\bar{n}_{out}(t)$ defined in Eq. (\ref{eq:nout6}) satisfies Eq. (\ref{eq:n}). The proof is completed.
\end{proof}

\subsection{Output intensity in steady state} \label{sec:photon-steady}

In this subsection, we compute the steady-state   output intensity when the system $G$ is driven by an $m$-channel photon input field $\vert\Psi_{\nu}\rangle$ in the class $\mathcal{F}_0$ (Eq. (\ref{class_F0})).

The following is the main result of this subsection, whose proof is given in the Appendix.

\begin{theorem}  \label{thm:n-out-2}
The steady-state  output intensity of the output fields of the system $G$ driven by the $m$-channel single-photon input field $\vert\Psi_{\nu}\rangle$ in the class $\mathcal{F}_0$ (Eq. (\ref{class_F0})) is given by
\begin{equation}
\label{eq:out}
\bar{n}_{out}(t)= \int_0^\infty g_{G^{+}}(r)^{\#}g_{G^{+}}(r)^{T}dr+%
\xi_{out}^{+}(t)^{\#}\xi^{+}_{out}(t)^{T}+\xi^{-}_{out}(t)^{\#}\xi^{-}_{out}(t)^{T} ,
\end{equation}
where
\begin{equation}\label{xi_out_pm}
\xi^-_{out}(t) =   \int_{-\infty}^\infty g_{G^-} (t-r) \xi_{in}^-(r) dr, \ \
\xi^+_{out}(t) =   \int_{-\infty}^\infty g_{G^+} (t-r) \xi_{in}^-(r)^\# dr.
\end{equation}
In particular, the total output intensity is given by   
\begin{equation}
\label{eq:tr1}
\mathrm{Tr}[\bar{n}_{out}(t)] = \sum_{j,k=1}^m \int_0^\infty |g_{G^+}^{jk} (t)|^2 dt
+ \sum_{j,k=1}^m | \xi_{out}^{+, jk}(t)|^2  +  \sum_{j,k=1}^m | \xi_{out}^{-, jk} (t)|^2,
\end{equation}
where $g_{G^+}^{jk}(t)$ is the element of $g_{G^+}$ on the $j$th row and $k$th column. The same applies to $g_{G^-}^{jk}(t)$, $\xi_{out}^{-, jk}$, and $\xi_{out}^{+, jk}$.
\end{theorem}

For the single input case, the steady-state output intensity is given by
\begin{equation}
\bar{n}_{out}(t)=\int_0^\infty \vert g_{G^{+}}(r)\vert^{2}dr+\left\vert \int_{-\infty}^\infty g_{G^{+}}(t-r)\xi_{in}^-(r)^\ast dr\right\vert ^{2}+\left\vert
\int_{-\infty}^\infty g_{G^{-}}(t-r)\xi_{in}^-(r)dr\right\vert^{2}.
\label{eq:out2}
\end{equation}

\begin{remark}
{\rm For the single input case, if $t_0 = 0$, in the Appendix we have shown that substitution of Eqs. (\ref{eq:interm1}), (\ref{eq:interm1b}), (\ref{eq:interm1d}) and (\ref{eq:interm1e}) into Eq. (\ref{eq:nout7}) yields
\begin{equation}
\bar{n}_{out}(t)=\int_0^t\left\vert g_{G^{+}}(r)\right\vert
^{2}dr+\left\vert \int_0^t(g_{G^{+}}(t-r)\xi^-(r)^\ast dr\right\vert ^{2}+\left\vert
\int_0^tg_{G^{-}}(t-r)\xi^-(r)dr\right\vert^{2}.  \label{eq:out2_temp}
\end{equation}
Assuming further that $C_+ = \Omega_+ = \Omega_- =0, C_- =\sqrt{\kappa_a}$,   and $\xi^-(t)$ is $\nu(t)$ in \cite[Eq. (18)]{MNM10}, then Eq. (\ref{eq:out2_temp}) reduces to \cite[Eq. (26)]{MNM10}.
}
\end{remark}

\subsection{Output covariance function in steady state}
\label{sec:photon-steady-covariance}

If the input field is a multichannel photon field state $\vert \Psi_\nu \rangle$ as defined in Eq. (\ref{eq:multichannel}), then it is easy to show that the input covariance function is
\begin{equation}\label{eq:R_photon_in}
R_{in}(t,r) = \mathbb{E}_{\Psi_\nu}[\breve{b}(t)\breve{b}^\dag(r)]= R_0(t-r) + \Delta(\xi_{in}^-(t),0)\Delta(\xi_{in}^-(r),0)^\dag ,
\end{equation}
where $\xi_{in}^-(t)$ is given in Eq. (\ref{xi_in_photon}) and $R_0$ is the vacuum covariance defined by (\ref{R_in_vacuum}). According to Theorem \ref{thm:covariance-transfer}, the output covariance function is
\begin{equation}\label{eq:covariance-transfer_photon}
R_{out}(t,r) = \int_{-\infty}^\infty g_{G}(t-\tau)   \mathrm{diag}(I_m, 0_m)
g_G(r-\tau)^\dag d\tau + \Delta(\xi_{out}^-(t), \xi_{out}^+(t))\Delta(\xi_{out}^-(r), \xi_{out}^+(r))^\dag ,
\end{equation}
with $\xi_{out}^-$ and $\xi_{out}^+$ given by Eq. (\ref{xi_out_pm}). Clearly
\begin{equation}
R_{in}(t,r) = R_{in}(r,t)^\dag, ~~ R_{out}(t,r) = R_{out}(r,t)^\dag .
\end{equation}

\emph{Example 1.} (Optical cavity) An optical cavity $G$ is a single open oscillator \cite{GZ00,BR04,WM08} with $\Omega_- =\omega \in \mathbb{R}$, $\Omega_+ = 0$, $C_- =\sqrt{\kappa }$, $C_+ =0$ \cite[section IV. B.]{GJN10a}). Let the pulse shape of the single-photon input field state $\vert 1_\nu\rangle$ be given by
\begin{equation} \label{eq:shape}
\nu (t)=\left\{
\begin{array}{ll}
\sqrt{2\gamma }e^{-\gamma t}, & t\geq 0, \\
0, & t<0.
\end{array}%
\right.
\end{equation}
The state  $\vert 1_\nu\rangle$ can describe a single-photon field emitted from an optical cavity with damping rate $\sqrt{2\gamma}$. The input covariance function is given by (\ref{eq:R-photon}). On the other hand by Eq. (\ref{xi_out_pm}),
\begin{equation} \label{xi_out_minus}
\xi_{out}^-(t)
= \left\{
\begin{array}{ll}
   \sqrt{2\gamma} e^{-\gamma t} -\frac{\kappa \sqrt{2\gamma}}{\frac{\kappa}{2}+i\omega-\gamma}\left(e^{-\gamma t}-e^{-(\frac{\kappa}{2}+i\omega)t}\right), & t \geq 0 \\
  0, & t <0
\end{array}
  \right. , ~~ \xi_{out}^+(t) \equiv 0 .
\end{equation}
Write
\begin{equation}
\chi(t,r) =  \left\{
 \begin{array}{ll}
   -\kappa e^{-(\frac{\kappa}{2}+i\omega)(t-r)}, & t>r \\
    \delta(t-r),     & t=r \\
   -\kappa e^{-(\frac{\kappa}{2}+i\omega)(r-t)},  & t<r
 \end{array}
 \right. .
\end{equation}
It can be checked that the steady-state output correlation function is thus
\begin{equation}
R_{out}(t,r) =\chi(t,r) \left[\begin{array}{cc}
                                   1 & 0 \\
                                   0 & 0
                                 \end{array}
 \right] + \Delta(\xi_{out}^-(t), 0)\Delta(\xi_{out}^-(r), 0)^\dag .
\end{equation}
Clearly, the mean input intensity is $\bar{n}_{in}(t) = \vert\nu(t)\vert^2$. By Theorem \ref{thm:n-out-2}, the steady-state mean output intensity is
$
\bar{n}_{out}(t) = |\xi_{out}^-(t)|^2 .
$
According to Theorem \ref{thm:n_out}, for $t \geq 0$,
\begin{equation}
m_-(t) = \frac{-\sqrt{2\kappa\gamma}}{\frac{k}{2}-\gamma-i\omega}\left[ \begin{array}{cc}
                                                                          e^{-\gamma t}-e^{-(\frac{\kappa}{2}-i\omega)t} & 0
                                                                        \end{array}
 \right], m_+(t) = \frac{-\sqrt{2\kappa\gamma}}{\frac{k}{2}-\gamma+i\omega}\left[\begin{array}{cc}
                                                                         0 & e^{-\gamma t}-e^{-(\frac{\kappa}{2}+i\omega)t}
                                                                       \end{array}
 \right].
\end{equation}
Hence, the evolution of system variables covariance $\Sigma_\nu(t)$ given in Eq. (\ref{eq:m}) can be expressed explicitly. In particular, when $\omega = 0$,
\begin{equation}
\Sigma_\nu(t) = \left[\begin{array}{cc}
                         1  &  0\\
                         0  &  0
                  \end{array}
 \right]+\frac{2 \gamma  \kappa }{(\frac{\kappa}{2} - \gamma )^2} (e^{-\frac{\kappa}{2} t}-e^{-\gamma t})^2 I_2,  ~~ t \geq 0.
\end{equation}

%

\emph{Example 2.} (Degenerate parametric amplifier) A degenerate parametric amplifier (DPA) is an open oscillator that is able to produce squeezed output fields \cite{GZ00,BR04,WM08}. A model for a DPA is \cite[pp. 220 and Chapter 10]{GZ00}
\begin{eqnarray}
\dot{\breve{a}}(t) & =& -\frac{1}{2}\left[
\begin{array}{cc}
\kappa & -\epsilon \\
-\epsilon & \kappa%
\end{array}
\right] \breve{a}(t)-\sqrt{\kappa}\breve{b}(t),  ~~ \breve{a}(0) = \breve{a}, \nonumber  \\
\breve{b}_{out}(t) & =& \sqrt{\kappa}\breve{a}(t)+\breve{b}(t), \ \  ( 0<\epsilon<\kappa ).
\end{eqnarray}
If the single-photon input $\vert 1_\nu \rangle$ has the pulse shape defined in Eq. (\ref{eq:shape}), then by Eq. (\ref{xi_out_pm}),
\begin{equation} \label{DPA_xi-}
\xi_{out}^-(t) = \left\{\begin{array}{ll}
                          \sqrt{2\gamma }e^{-t \gamma }+ \frac{2\kappa \sqrt{2\gamma} e^{-\frac{\kappa}{2}t} \left(e^{-\frac{2\gamma -\kappa}{2} t} (2 \gamma -\kappa )-(2\gamma-\kappa ) \cosh\frac{t\epsilon}{2}+\epsilon\sinh\frac{t\epsilon}{2}\right)}{(\kappa -2 \gamma -\epsilon ) (\kappa -2 \gamma +\epsilon )} , & t \geq 0 \\
                          0, & t<0
                        \end{array} \right.
 ,
\end{equation}
\begin{equation} \label{DPA_xi+}
\xi_{out}^+(t) =\left\{\begin{array}{ll}
 \frac{2 \kappa \sqrt{2\gamma} e^{-\frac{t\kappa}{2}}\left(-\epsilon e^{-\frac{2\gamma-\kappa}{2}t} +\epsilon  \cosh\frac{t\epsilon}{2}-(2\gamma-\kappa) \sinh\frac{t\epsilon}{2}\right)}{(\kappa -2\gamma -\epsilon) (\kappa -2 \gamma +\epsilon )}, & t\geq 0 \\
 0, & t<0
                        \end{array} \right. .
\end{equation}

When the system is initialized at time $t_0 = 0$, by Eq. (\ref{eq:out2_temp}), the mean output intensity is
\begin{equation}
\bar{n}_{out}(t)  =\kappa^{2}\int_{0}^t e^{-\kappa r}\sinh^{2}%
\frac{\epsilon r}{2}dr + \vert \xi_{out}^-(t) \vert^{2}+\vert \xi_{out}^+(t) \vert^{2}, ~~ t \geq 0 .
\end{equation}

Since $\vert \Psi_{in}\rangle = B^\ast(\nu)\vert 0\rangle$, the covariance function corresponding to $\vert \Phi_{R_{in}}\rangle$ (the vacuum field state) is $R_0(\tau)$ in Eq. (\ref{R_in_vacuum}). According to Theorem \ref{thm:covariance-transfer}, the output covariance function is
\begin{equation} \label{eq:dpa_R_out}
    R_{out,g}(\tau) = \left\{ \begin{array}{ll}
                                \kappa R_a(\tau) -\kappa e^{A\tau}\mathrm{diag}(1,0)  +\mathrm{diag}(1,0)\delta(\tau),  & \tau > 0 \\
                                \kappa R_a(\tau) -\frac{1}{2}\kappa e^{A\tau}\mathrm{diag}(1,0)  - \frac{1}{2}\kappa \mathrm{diag}(1,0)  e^{A^\dag \tau} + \mathrm{diag}(1,0) \delta(\tau), & \tau = 0 \\
                                \kappa R_a(\tau) -\kappa \mathrm{diag}(1,0)  e^{-A^\dag \tau}+\mathrm{diag}(1,0) \delta(\tau),  & \tau < 0
                              \end{array}
      \right. ,
\end{equation}
where
\begin{equation}
R_a(\tau) = \left\{\begin{array}{ll}
                     e^{A\tau}\Upsilon , & \tau >0  \\
                     \Upsilon , & \tau = 0  \\
                     \Upsilon e^{-A^\dag \tau} , & \tau < 0
                   \end{array}
  \right.,  ~~~ \Upsilon =\kappa  \int_0^\infty   e^{At}\left[\begin{array}{cc}
                                              1 & 0 \\
                                              0 & 0
                                            \end{array}
 \right]  e^{A^\dag t}dt .
\end{equation}
As a result, according to Eq. (\ref{eq:covariance-transfer_photon}), the covariance function of the output field of the DPA driven by the single-photon input field $\vert 1_\nu \rangle$ can be expressed as
\begin{equation}
R_{out}(t,r) = R_{out,g}(t-r) + \Delta(\xi_{out}^-(t), \xi_{out}^+(t))\Delta(\xi_{out}^-(r), \xi_{out}^+(r))^\dag ,
\end{equation}
where $R_{out,g}(t-r)$ is given in Eq. (\ref{eq:dpa_R_out}).

\begin{remark}
In the cavity example, the output correlation function has the same form as the single photon correlation function (\ref{eq:R-photon}), while this is not the case for the degenerate parametric amplifier. Thus the steady-state output state may in some cases be a single photon state, while in other cases more complex states may result. A general class of output states is the subject of the following section.
\end{remark}

\section{Photon-Gaussian states}
\label{sec:state-photon}

Since the states produced as the outputs of linear quantum systems need not necessarily be photon states, we consider a larger class of states that has the property that if a state in this class is input to a linear quantum system, then in steady state the output state is also in this class. This class of states is defined and studied in Section \ref{sec:pgS}. However, the expressions for specifying these states are quite complicated, and so in Section \ref{sec:photon-steady-output-state} we consider the simpler single channel case for pedagogical reasons. The calculations involved in determining the output states make use of the stable, but non-causal, inversions discussed in Section \ref{sec:identity}.

\subsection{The one channel case} \label{sec:photon-steady-output-state}

Given an initial joint system-field state $\rho_{0g}=\left\vert \phi\right\rangle \left\langle \phi\right\vert
\otimes\left\vert 0\right\rangle \left\langle 0\right\vert$, denote
\begin{equation}\label{eq:rho_inf_g}
\rho_{\infty g} = \lim_{t\rightarrow\infty,t_{0}\rightarrow-\infty}U\left(
t,t_{0}\right)  \rho_{0g}U\left(t,t_{0}\right) ^{\ast}.
\end{equation}
Define
\begin{equation}\label{eq:rho_field}
\rho_{field,g}=\mathrm{Tr}_{sys}[\rho_{\infty g}],
\end{equation}
where the subscript ``sys'' indicates that the trace operation is with respect to the system.  According to Theorem \ref{thm:spectral-transfer},  $\rho_{field,g}$ is the steady-state output field density with covariance function $R_{out}[i\omega]$ given in Eq. (\ref{eq:R-transfer}).

We are in a position to prove a result concerning the output field of quantum linear systems driven by single-photon states.

\begin{proposition}   \label{prop:out-state-photon}  
Let $m=1$ and suppose the input state $\rho_{in} = \vert 1_\nu \rangle\langle 1_\nu \vert$ is a single photon state. Then the steady-state output field state for the linear quantum system $G$ is given by
\begin{equation}
\rho_{out} = ( B^\ast( \xi^-_{out} ) - B(\xi^+_{out}) ) \rho_{field,g} ( B^\ast( \xi^-_{out} ) - B(\xi^+_{out}) )^\ast
\label{eq:out-state-photon}
\end{equation}
where
\begin{equation}
\Delta ( \xi^-_{out}[s] , \xi^+_{out}[s] ) =  \Xi_{G}[s] \Delta ( \nu[s] , 0 ) ,
\end{equation}
and $\rho_{field,g}$, defined in Eq. (\ref{eq:rho_field}), is the steady-state density operator for the output field with zero mean and covariance function
\begin{equation} \label{eq:R_out}
R_{out}[i\omega] =  \Xi_G[ i\omega] R_{in}[i\omega]  \Xi_G[i\omega]^\dag ,
\end{equation}
where $R_{in}[i\omega]$ is the Fourier transform of $R_{0}(\tau)$ defined in Eq. (\ref{R_in_vacuum}).
\end{proposition}


\begin{proof}
The initial joint system-field density is
\[
\rho_{0}=\left\vert \phi\right\rangle \left\langle \phi\right\vert
\otimes\left\vert 1_{\nu}\right\rangle \left\langle 1_{\nu}\right\vert
=B^{\ast}(\nu)\rho_{0g}B(\nu).
\]
 The steady-state joint system-field density, denoted by $\rho_{\infty}$,  is%
\begin{align}
\rho_{\infty} &  =\lim_{t\rightarrow\infty,t_{0}\rightarrow-\infty}U\left(
t,t_{0}\right)  \rho_{0}U\left(  t,t_{0}\right)  ^{\ast} \label{rho_infty}\\
&  =\lim_{t\rightarrow\infty,t_{0}\rightarrow-\infty}U\left(  t,t_{0}\right)
B^{\ast}(\nu)\rho_{0g}B(\nu)U\left(  t,t_{0}\right)  ^{\ast} \nonumber \\
&  =\lim_{t\rightarrow\infty,t_{0}\rightarrow-\infty}U\left(t,t_{0}\right)
B^{\ast}(\nu)U\left(t,t_{0}\right)^{\ast}\rho_{\infty,g}\lim_{t\rightarrow\infty,t_{0}\rightarrow-\infty
}U\left(  t,t_{0}\right)  B(\nu)U\left(  t,t_{0}\right)^{\ast}, \nonumber
\end{align}%
where $\rho_{\infty,g}$ is given in Eq. (\ref{eq:rho_inf_g}). Now we find expressions for the other terms.
\begin{eqnarray}
&& \lim_{t\rightarrow\infty,t_{0}\rightarrow-\infty}U\left(t,t_{0}\right)
B^{\ast}(\nu)U\left(t,t_{0}\right)^{\ast} \nonumber   \\
&=& \lim_{t\rightarrow\infty,t_{0}\rightarrow-\infty}U\left(t,t_{0}\right)
\int_{-\infty}^\infty \nu(r) I_{sys} \otimes b^\ast(r)dr U\left(t,t_{0}\right)^{\ast} \nonumber  \\
&=&  \lim_{t\rightarrow\infty,t_{0}\rightarrow-\infty}U\left(t,t_{0}\right)
\int_{t_0}^t \nu(r) I_{sys} \otimes b^\ast(r)dr U\left(t,t_{0}\right)^{\ast} \nonumber \\
&=& I_{sys} \otimes \int_{-\infty}^\infty  \nu(r) b^{-\ast}(r,-\infty) dr
\end{eqnarray}
where $b^-(t,t_0) = U(t,t_0) b(t) U^\ast(t,t_0)$.
Now
\begin{equation}
\breve b_{out}(t) = U^\ast(t,t_0) \breve b(t) U(t,t_0)
=Ce^{A(t-t_0)}\breve{a} + \int_{t_0}^t g_{G}(t-r) \breve b(r) dr
\end{equation}
and so
\begin{equation}
\breve b(t) = Ce^{A(t-t_0)}\breve{a}+  \int_{t_0}^t g_{G}(t-r) \breve b^-(r,t_0) dr.
\end{equation}
Now send $t_0 \to -\infty$ to obtain
\begin{eqnarray}
\breve b(t) &=& \int_{-\infty}^t g_{G}(t-r) \breve b^-(r,-\infty) dr =  \int_{-\infty}^\infty g_{G}(t-r) \breve b^-(r,-\infty) dr .
\end{eqnarray}
Next, using the stable inverse of $G$ (Lemma \ref{lem:G_inv}, Eq. (\ref{eq:G-inv})) this implies
\begin{eqnarray}
\breve b^-(t,-\infty) &=& \int_{-\infty}^\infty g_{G^{-1}} (t-r) \breve b (r)dr
\end{eqnarray}
Therefore
\begin{equation}\label{eq:b_out_ss}
 \lim_{t\rightarrow\infty,t_{0}\rightarrow-\infty}U\left(t,t_{0}\right)
B^{\ast}(\nu)U\left(t,t_{0}\right)^{\ast}  = I_{sys}\otimes ( B^\ast( \xi^-_{out} ) - B(\xi^+_{out}) ),
\end{equation}
using Eq. (\ref{eq:G-inv}) and the definitions of $\xi^-_{out}(t)$ and $\xi^+_{out}(t)$ in Eq. (\ref{xi_out_pm}).

By Eqs. (\ref{rho_infty}) and (\ref{eq:b_out_ss}),
\begin{equation}
\rho_{\infty} = (I_{sys}\otimes ( B^\ast( \xi^-_{out} ) - B(\xi^+_{out}) ))\rho_{\infty,g} (I_{sys}\otimes ( B^\ast( \xi^-_{out} ) - B(\xi^+_{out}) ))^\ast
\end{equation}
This, together with Eq.  (\ref{eq:rho_field}),
\begin{eqnarray}
\rho_{out} & =&  \mathrm{Tr}_{sys}[\rho_{\infty} ] \nonumber\\
& =& \mathrm{Tr}_{sys}[(I_{sys}\otimes ( B^\ast( \xi^-_{out} ) - B(\xi^+_{out}) ))\rho_{\infty,g} (I_{sys}\otimes ( B^\ast( \xi^-_{out} ) - B(\xi^+_{out}) ))^\ast] \nonumber \\
&=&( B^\ast( \xi^-_{out} ) - B(\xi^+_{out}) ) \rho_{field,g} ( B^\ast( \xi^-_{out} ) - B(\xi^+_{out}) )^\ast. \nonumber
\end{eqnarray}
Eq. (\ref{eq:out-state-photon}) is thus established.
\end{proof}

The output state $\rho_{out}$  given by (\ref{eq:out-state-photon}) is determined by a matrix $\xi_{out} = \Delta ( \xi^-_{out} , \xi^+_{out} ) $ of functions which is obtained by convolving an input matrix $\xi_{in} = \Delta ( \nu , 0) $ with the system $G$, and a Gaussian state $\rho_{field,g}$ whose covariance $R_{out}$ is given by the usual transfer relation (\ref{eq:R-transfer}) where $R_{in}$ is the covariance function for the vacuum field.
It will be shown by Proposition \ref{prop:normalization} in Section \ref{sec:state-photon} that $\rho_{field,g}$ in Eq. (\ref{eq:out-state-photon}) is indeed normalized.

\emph{Example 3.} Refer to \emph{Example 1} on optical cavities. By Eq. (\ref{R_out_vac}), $\rho_{field,g}$ is a vacuum state $|0\rangle\langle0|$. By Proposition \ref{prop:out-state-photon}, the steady-state output field state is a pure state $\left\vert\Psi_{out}\right\rangle = B^\ast(\xi_{out}^-)\vert0\rangle$, where $\xi_{out}^-$  is given in Eq. (\ref{xi_out_minus}). Clearly, the output is in a single-photon state.

\emph{Example 4.} Refer to \emph{Example 2} on degenerate parametric amplifiers. According to Proposition \ref{prop:out-state-photon}, the steady-state output state is $\rho_{out} = (B^\ast(\xi_{out}^-)- B(\xi_{out}^+))\rho_{field,g}(B^\ast(\xi_{out}^-)- B(\xi_{out}^+))^\ast$, where $\xi_{out}^-, \xi_{out}^+$ are given in Eqs. (\ref{DPA_xi-})-(\ref{DPA_xi+}), and the covariance function $R_{out}$ of $\rho_{field,g}$  is given by Eq. (\ref{eq:dpa_R_out}). The normalization condition for $\rho_{out}$  will be given by Eq. (\ref{eq:F_1}). It can be verified that indeed  $\mathrm{Tr}[\rho_{out}] =1$. Clearly, $\rho_{out}$ is not a single-photon state. 

\subsection{Photon-Gaussian States}
\label{sec:pgS}

In the previous section, in particular \emph{Example 4}, we saw that a quantum linear system produces a somewhat complicated output state $\rho_{out}$ from a single photon input state.  This output state was determined by pulse shapes and a Gaussian state. In this section we abstract the form of this output state and define a class $\mathcal{F}$ of {\em photon-Gaussian} states. The class $\mathcal{F}$ contains single-photon states studied in \cite{Milburn08} and \cite{ZJ11b} as special cases.
Furthermore, in Theorem \ref{thm:main} we show that this class of states is invariant under the steady-state action of a linear quantum system, that is, $\rho_{in} \in \mathcal{F}$ implies $\rho_{out} \in \mathcal{F}$.

We first introduce some notations.  Given $t_1,\ldots, t_j \in \mathbb{C}$ and $\xi_1,\ldots,\xi_j \in L^2(\mathbb{C},\mathbb{C}^{2m})$ where $j$ is an arbitrary positive integer, define
\begin{eqnarray}
M_{\xi}^-(t_{1\rightarrow j}) &:=&  \xi_1(-t_1)\otimes_{c}\cdots\otimes_{c}\xi_j(-t_j) , \nonumber \\
M_{\xi}(t_{1\rightarrow j}) &:=&  \xi_1(t_1)\otimes_{c}\cdots\otimes_{c}\xi_j(t_j) , \nonumber \\
M_{\xi}^+(t_{1\rightarrow j}) &:=&  \xi_j(t_1)\otimes_{c}\cdots\otimes_{c}\xi_1(t_j) ,  \label{eq:M_xi}
\end{eqnarray}
where $\otimes_{c}$ is the Kronecker product. Similarly, for the operators $\breve{b}(t)$, define
\begin{equation}\label{eq:M_b}
    M_{\breve{b}}(t_{1\rightarrow j}) :=  \breve{b}(t_1)\otimes_{c}\cdots\otimes_{c}\breve{b}(t_j) .
\end{equation}
Finally for a matrix $A$, let $A^{\otimes_c^j} := A\otimes_{c}\cdots\otimes_{c}A$
be a $j$-way Kronecker tensor product. Clearly, when $j=1$, $A^{\otimes_c^1} = A$. Define $\xi(t) := \Delta(\xi^-(t), \xi^+(t)) \in \mathbb{C}^{2m \times 2m}$
with entries $\xi_{jk}^-$ and $\xi_{jk}^+$ ($j,k=1,\ldots,m$) for matrices $\xi^-(t), \xi^+(t)$ respectively. Let $\rho_R$ be a zero mean stationary Gaussian field state with correlation function $R(\tau)$.

The following equation will be used in Definition \ref{def:F}.
\begin{equation}\label{eq:innerproduct}
\underbrace{\int_{-\infty}^\infty\cdots\int_{-\infty}^\infty}_{2m} (M_{\xi}^+(t_{1\rightarrow m})^\#\otimes_{c}M_{\xi}(t_{m+1\rightarrow 2m}))^T J^{\otimes_{c}^m}\otimes_{c} \Theta^{\otimes_{c}^m}\mathrm{Tr}[\rho_R M_{\breve{b}}(t_{1\rightarrow 2m})] dt_1\ldots dt_{2m} =1,
\end{equation}
where $\rho_R$ is a zero-mean Gaussian state with covariance function $R$.

\begin{definition}\label{def:F}
A state $\rho_{\xi, R}$ is said to be a \emph{photon-Gaussian} state if it belongs to the set
\begin{eqnarray}
\mathcal{F} &:=& \left\{\rho_{\xi, R} = \prod\limits_{k=1}^{m}\sum_{j=1}^{m}\left(B_j^\ast (\xi_{jk}^-) -B_j(\xi_{jk}^+) \right)\rho_R\left(\prod\limits_{k=1}^{m}\sum_{j=1}^{m}\left(B_j^\ast (\xi_{jk}^-) -B_j(\xi_{jk}^+) \right)\right)^\ast \right. \nonumber\\
& & \ \ \ \  \ \ \ \ \ \ \  \left.  :  \xi \mathrm{~and~} \rho_R \mathrm{~satisfy ~ Eq.~} (\ref{eq:innerproduct})   \right\}. \label{class_F}
\end{eqnarray}
\end{definition}


\begin{proposition}\label{prop:normalization}
The photon-Gaussian states   $\rho_{\xi, R} \in \mathcal{F}$ are normalized: $\mathrm{Tr}[\rho_{\xi, R}] = 1$.
\end{proposition}

\begin{proof}
Partition $\xi^-,\xi^+,\xi$ to be
\begin{equation}
\xi^- = [\xi^{-,1} \ \cdots \ \xi^{-,m}], \ \xi^+ = [\xi^{+,1} \ \cdots \ \xi^{+,m}], \ \xi = [\xi^1 \ \cdots \ \xi^{2m}] .
\end{equation}
It is not hard to show that $\rho_{\xi,R} \in \mathcal{F}$ is in the form of
\begin{align}
&\mathrm{Tr}[\rho_{\xi, R}] \label{eq:psi_in_7} \\
=&  \mathrm{Tr}\left[\prod_k\int_{-\infty}^\infty \xi^k(t)^T \Theta\breve{b}(t) dt \rho_R  \left(\prod_k\int_{-\infty}^\infty \xi^k(t)^T \Theta\breve{b}(t) dt\right)^\ast\right] \nonumber \\
=& \mathrm{Tr}\left[\rho_R \left(\underbrace{\int_{-\infty}^\infty\cdots\int_{-\infty}^\infty}_m M_{\xi}^+(t_{1\rightarrow m})^\dag J^{\otimes_c^m}M_{\breve{b}}(t_{1\rightarrow m})  dt_1\cdots dt_m\right)^\ast  \right. \nonumber \\
 &  \ \ \ \left. \times \underbrace{\int_{-\infty}^\infty\cdots\int_{-\infty}^\infty}_m M_{\xi}(t_{1\rightarrow m})^T \Theta^{\otimes_c^m}M_{\breve{b}}(t_{1\rightarrow m})dt_1\cdots dt_m  \right] \nonumber\\
=&\underbrace{\int_{-\infty}^\infty\cdots\int_{-\infty}^\infty}_{2m} (M_{\xi}^+(t_{1\rightarrow m})^\#\otimes_{c}M_{\xi}(t_{m+1\rightarrow 2m}))^T J^{\otimes_{c}^m}\otimes_{c} \Theta^{\otimes_{c}^m}\mathrm{Tr}[\rho_R M_{\breve{b}}(t_{1\rightarrow 2m})] dt_1\ldots dt_{2m},  \nonumber
\end{align}
where $\Theta$ is as introduced in the Notation part of the Introduction section.
Thus, that $\mathrm{Tr}[\rho_{\xi,R}] =1$ is equivalent to that Eq. (\ref{eq:innerproduct}) holds. The proof is completed.
\end{proof}

Note that when $m=1$, a state $\rho_{\xi,R} \in \mathcal{F}$ is
\begin{equation}\label{eq:Psi_2}
\rho_{\xi,R} = ( B^\ast(\xi^-)- B(\xi^+))\rho_R ( B^\ast(\xi^-)- B(\xi^+))^\ast ,
\end{equation}
and Eq. (\ref{eq:innerproduct}) reduces to
\begin{equation}\label{eq:F_1}
\int_{-\infty}^\infty \int_{-\infty}^\infty [\xi^-(t)^\ast \ -\xi^+(t)] \mathrm{Tr}[\rho_R \breve{b}(t)\breve{b}^\dag(r)] \left[ \begin{array}{c}
                                                            \xi^-(r) \\
                                                            -\xi^+(r)^\ast
                                                          \end{array}
 \right]dtdr = 1 .
\end{equation}

We are now ready to state the main result of this section. The proof is given in the Appendix.

\begin{theorem}   \label{thm:main}
 Let  $\rho_{\xi_{in}, R_{in}}  \in \mathcal{F}$ be a photon-Gaussian input state. Then the linear quantum system $G$ produces in steady state a photon-Gaussian output state
$\rho_{\xi_{out}, R_{out}} \in \mathcal{F}$, where
\begin{eqnarray}
 \xi_{out}[s]  &=&  \Xi_{G}[s] \xi_{in}[s] ,  \label{eq:xi_out} \\
R_{out}[i\omega] &=&  \Xi_G[ i\omega] R_{in}[i\omega]  \Xi_G[i\omega]^\dag  . \label{eq:R_out_gnr}
\end{eqnarray}
\end{theorem}
Without confusion, we may use the shorthand $\rho_{out}$ for $\rho_{\xi_{out}, R_{out}}$.

\begin{remark} \label{rem:pure}
{
When the input state is a pure state and the system is passive (e.g., an optical cavity or a beamsplitter), it can be seen from Theorem \ref{thm:main} that in steady state the output field is in a pure state. That is
\[
\rho_{\xi_{out}, R_{out}}  = \vert\Psi_{\xi_{out}, R_{out}}\rangle\langle\Psi_{\xi_{out}, R_{out}}\vert.
\]
In this case we use $\vert\Psi_{\xi_{out}, R_{out}}\rangle$ to denote the steady-state output field state. Again, without confusion, we may use the shorthand $|\Psi_{out}\rangle$ for $\vert\Psi_{\xi_{out}, R_{out}}\rangle$.
}
\end{remark}

We remark that the Gaussian part of the specification of photon-Gaussian states  is needed to allow for quantum linear systems with active elements, such as degenerate parametric amplifiers, \cite{GZ00,WM08,WM10}. In general, while passive devices produce vacuum from vacuum, active devices produce nontrivial Gaussian states from vacuum, and photon-Gaussian states from single photon states, as show in Example 3.

This result provides a complete description of how quantum linear systems process highly non-classical photon-Gaussian states. In particular, the result provides the output response to single photon inputs. The result may be used for dynamical analysis, or for synthesis. Indeed, one could contemplate generalizations of the synthesis results given in \cite{Milburn08} and in Section \ref{sec:synthesis}  to the class of photon-Gaussian states.

\emph{Example 5.} (Beam splitter) Beam splitters are archetype of static and passive quantum optical instruments \cite{BR04}, which can be modeled as
$b_{out}(t) =  S_- b(t)$ where
\begin{equation}\label{BS}
S_- = \left[  \begin{array}{cc}
                \sqrt{\eta} & \sqrt{1-\eta}  \\
                -\sqrt{1-\eta} & \sqrt{\eta}
              \end{array}
  \right]
\end{equation}
with $0\leq \eta \leq 1$. Let $\xi_{in}^- = \mathrm{diag}(\nu_1, \nu_2), ~~ \xi_{in}^+ = 0$. According to Theorem \ref{thm:main} and Remark \ref{rem:pure},
\begin{eqnarray}
&& \vert\Psi_{out}\rangle   \nonumber \\
&=& \sqrt{\eta (1-\eta)} \int_{-\infty}^{\infty} \nu_1   (t) b_1^* (t)dt \int_{-\infty}^{\infty} \nu_2   (t) b_1^* (t)dt |0_1\rangle \otimes  |0_2\rangle   \nonumber\\
&& + \eta \int_{-\infty}^{\infty} \nu_1   (t) b_1^* (t)dt  |0_1\rangle \otimes  \int_{-\infty}^{\infty} \nu_2   (t) b_2^* (t)dt |0_2\rangle  \nonumber \\
&& - (1-\eta) \int_{-\infty}^{\infty} \nu_2   (t) b_1^* (t)dt |0_1\rangle \otimes \int_{-\infty}^{\infty} \nu_1   (t) b_2^* (t)dt |0_2\rangle   \nonumber \\
& & -\sqrt{\eta (1-\eta)} |0_1\rangle \otimes  \int_{-\infty}^{\infty} \nu_1   (t) b_2^* (t)dt  \int_{-\infty}^{\infty} \nu_2   (t) b_2^* (t)dt  |0_2\rangle . \label{eq:BS_Psi}
\end{eqnarray}
In particular, when $\eta = 1/2$ and $\nu_1   (t) \equiv \nu_2   (t), \forall t \in \mathbb{R}$, the steady-state output state is
\begin{equation}
\vert\Psi_{out}\rangle =  \frac{1}{2} \left ( \int_{-\infty}^{\infty} \nu_1   (t) b_1^* (t)dt \right)^2 |0_1\rangle \otimes  |0_2\rangle - \frac{1}{2} |0_1\rangle \otimes  \left ( \int_{-\infty}^{\infty} \nu_1   (t) b_2^* (t)dt \right)^2  |0_2\rangle .
\end{equation}%
In this case, the two photons can not exit from distinct output arms of the beam splitter.
These results are consistent with the results of the  calculations in \cite[Sec. 16.4.2]{WM08}.

\emph{Example 6.} (Linear quantum systems driven by both a single-photon state and a coherent state)  The methods used in this paper may easily be adapted to treat multichannel input fields where some channels are coherent states while others are single photons. Consider a linear quantum system driven by a single-photon state $\vert 1_\nu\rangle$ (channel 1) and a coherent state $\vert\alpha\rangle$ (channel 2) simultaneously, Fig.~\ref{fig:pc}.

\begin{figure}
\centering
\includegraphics[width=2.5in]{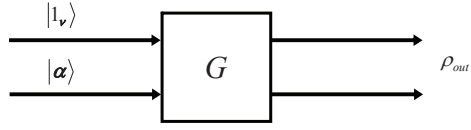}\\
\caption{A quantum linear stochastic system $G$ driven by a single-photon state $\vert 1_\nu\rangle$ and a coherent state $\vert\alpha\rangle$ simultaneously.}
\label{fig:pc}
\end{figure}

Denote the pure input state by $
\vert\Psi\rangle = \vert 1_\nu\rangle \otimes \vert\alpha\rangle = B_1^\ast(\nu) \vert 0\rangle \otimes \vert\alpha\rangle $. It is easy to show that the correlation function for the composite state $\vert 0\rangle \otimes \vert\alpha\rangle$ is
\begin{equation}\label{R_input_2}
R_{in}(t,r) =\langle 0\alpha\vert \breve{b}(t)\breve{b}^\dag(r) \vert 0\alpha\rangle
=\left[
\begin{array}{cccc}
  \delta(t-r) & 0                                           & 0 & 0 \\
  0               & \delta(t-r)+\alpha^\ast(r)\alpha(t) & 0 & \alpha(t)\alpha(r) \\
  0               & 0                                           & 0 & 0 \\
  0               & \alpha^\ast(t)\alpha^\ast(r)            & 0 & \alpha^\ast(t)\alpha(r)
\end{array}
\right] .
\end{equation}
Define
\begin{equation}\label{tilde_xi}
\xi_{out}^{-,j}(t) =\int_{-\infty}^\infty g_{G^-}^{j1}(t-r)\nu(r)dr, ~~ \xi_{out}^{+,j}(t) = \int_{-\infty}^\infty g_{G^+}^{j1}(t-r)\nu(r)^\ast dr, ~~ j=1,2.
\end{equation}
By Theorem \ref{thm:main}, it is not hard to show that the steady-state output state of the system driven by $\vert\Psi\rangle$ is
\begin{equation}\label{tilde_Psi}
\rho_{out} = \left( \sum_{j=1}^2 (B_j^\ast(\xi_{out}^{-,j})-B_j(\xi_{out}^{+,j}))\right) \rho_{field,g}  \left( \sum_{j=1}^2 (B_j^\ast(\xi_{out}^{-,j})-B_j(\xi_{out}^{+,j}))\right)^\ast,
\end{equation}
where  $\rho_{field,g}$ has covariance function $R_{out}$ given by Eq. (\ref{eq:covariance-transfer}).

\section{Photon  shape synthesis} \label{sec:synthesis}

In this section we generalize the result of photon wavepacket  shape synthesis results presented in \cite{Milburn08}. We consider a class of passive linear quantum systems for which $C_{+}=0$ and $\Omega _{+}=0$. From Eq. (\ref{system-a}) it is easy to see that%
\begin{equation}
A^{\dag}+A+C^{\dag}C =0, \ S^{\dag}C+(BS)^{\dag}=0.
\end{equation}%
In this case $\Xi_G[s]$ defined in Eq. (\ref{eq:G_omega}) is $\Xi_G[s]=\Delta(\Xi_{G^-}[s],0)$. By Proposition \ref{prop:identity},
\begin{equation}
\Xi_G[-s^\ast]^\dag\Xi_G[s]=I, \ \forall s \in \mathbb{C} .
\end{equation}%
That is, $\Xi_G[s]$ is \emph{all-pass} \cite[pp. 357]{ZDG96}.

In what follows we study the following {\em pulse shaping} problem. Given two pulse shapes $\nu(t)$ and $\nu_{out}(t)$ satisfying $\int_{-\infty}^\infty \vert\nu(t)\vert^2 dt = \int_{-\infty}^\infty \vert\nu_{out}(t)\vert^2 dt=1$, let %
$\nu_{out}[s]$ and $\nu[s]$ be the Laplace transform of $\nu_{out}(t)$ and $\nu(t)$ respectively. If $\nu_{out}[s]$ and $\nu[s]$ satisfy
\begin{equation}\label{eq:nu_xi}
|\nu_{out}[i\omega]|^2 = |\nu[i\omega]|^2,  \ \forall \omega \in \mathbb{R},
\end{equation}
we show that, under mild conditions, there is an all-pass system which maps the input state $\vert 1_{\nu}\rangle$  to produce the output field state  $\vert 1_{\nu_{out}}\rangle$.

Define
\begin{equation}
\Xi_{G_d}[s] = \frac{\nu_{out}[s]}{\nu[s]} .
\end{equation}

We make the following assumptions on $\Xi_{G_d}[s]$:

\textbf{Assumptions}
\begin{enumerate}
\item $\Xi_{G_d}[s]$ is real-rational and has a state-space realization \cite[Chapter 3]{ZDG96}
\begin{equation}
\Xi_{G_d}[s] = \left[
\begin{array}{c|c}
A & B \\ \hline
C & D%
\end{array}
\right] .
\end{equation}

\item $\Xi_{G_d}[s]$ is Hurwitz stable.

\item The above state space realization is minimal.

\item $D=I$.
\end{enumerate}

\begin{theorem} \label{thm:shape}
If the given input and output  pulse shapes satisfy Eq. (\ref{eq:nu_xi}), and $\Xi_{G_d}[s]$
satisfies Assumptions 1) - 4), then there exists a linear quantum stochastic system $\Xi_G[s]=\Delta(\Xi_{G_d}[s],0)$
solving the pulse shaping problem:
\begin{equation}
S_{-}=I,C_{-}=C,~C_{+}=0,~\Omega_{-}=\frac{i}{2}\left( XA-A^{\dag
}X\right) ,~\Omega_{+}=0.
\label{eq:pulse-shaper}
\end{equation}
\end{theorem}
\begin{proof}
By Assumption 2, there exits a matrix $X = X^\dag \geq0$ such that%
\begin{equation}
A^{\dag}X+XA+C^{\dag}C=0.  \label{eq:obsv}
\end{equation}
In fact, by Assumption 3, $X>0$. Consequently, by Corollary 13.30 in \cite{ZDG96},
\begin{equation} \label{eq:passivity}
D^{\dag}C+B^{\dag}X  =0,   ~~  D^{\dag}D  =I.
\end{equation}
Eqs. (\ref{eq:obsv}) and (\ref{eq:passivity}) can be rewritten as%
\begin{equation}
X^{-1}A^{\dag}+AX^{-1}+BB^{\dag}  =0,  ~~ B = -X^{-1}C^{\dag} .
\end{equation}
Then by Theorem 5.1 in \cite{MP11a}, $\Xi_{G_d}[s]$ can be implemented by an all-pass  linear
quantum stochastic system with parameters given in equation (\ref{eq:pulse-shaper}).
This completes the proof.
\end{proof}

\begin{figure}
\centering
\includegraphics[width=4in]{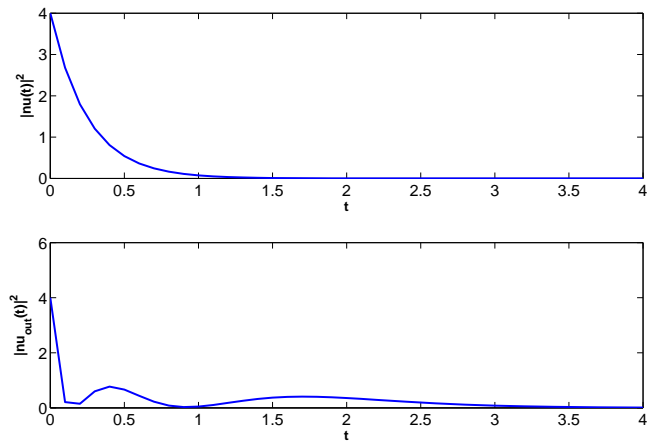}\\
\caption{The upper plot is for $\nu(t)$ and the lower plot is for $
\nu_{out}(t)$. The horizontal axes are time $t$, while the vertical axes are $\vert \nu(t) \vert^2$ (the upper) and $\vert \nu_{out}(t) \vert^2$  (the lower) respectively.}
\label{fig:shape}
\end{figure}

\emph{Example 7. }  Consider the following two functions $\nu$ given by Eq. (\ref{eq:shape})  and
\begin{equation}\label{exam:nu_out}
   \nu_{out}(t) = \left \{ \begin{array}{ll}
                 2 \left(\frac{54+12i}{5}
   e^{-3 t}-5(2+i) e^{-2 t}+\frac{1+13i}{5}
   e^{(-1-i) t}\right), &  t \geq 0 \\
                 0, & t<0
               \end{array}
               \right. .
\end{equation}
It can be checked that $\nu$ and $\nu_{out}$ satisfy Eq. (\ref{eq:nu_xi}). The shapes of $\nu(t)$ and $\nu_{out}(t)$ are plotted in Fig.~\ref{fig:shape}.  Note that
\begin{equation}
\frac{\nu_{out}[s]}{\nu[s]}=\frac{s-3}{s+3}\bullet \frac{%
s-(1-i)}{s+(1+i)}.
\end{equation}
In the language of $(S_-,C_-,C_+,\Omega_-,\Omega_+)$ as discussed in Sec.~\ref{sec:lqss}, by Theorem \ref{thm:shape}, the transfer function $\frac{s-3}{s+3}$ corresponds to $G_{1}=\left(1,\sqrt{6},0,0,0\right) $, while
the transfer function $\frac{s-(1-i)}{s+(1+i)}$ corresponds $%
G_{2}=\left( 1,\sqrt{2},0,1,0\right) $. Therefore, the whole system that transfers $\nu$ to $\nu_{out}$ is a cascaded system made of $G_1$ and $G_2$ \cite[Definition 5.3]{GJ09}.

\section{Conclusions}
 \label{sec:conclusion}

In this paper we have investigated the response of linear quantum systems driven by multi-channel photon input fields. Results concerning  the intensity and correlations  and states of output fields have been presented. In particular we have defined the class of photon-Gaussian states which arise when quantum optical systems with active components are driven by multi-channel photon input fields. Examples from quantum optics have been used to illustrate the results presented. Future work will include application of the results to specific problems in quantum technology.

\section*{Acknowledgment}

The authors wish to thank H. Nurdin for his very helpful discussions and suggestions.

\appendix

\section*{Proofs} \label{sec:appendix}

\emph{Proof of Theorem \ref{thm:n-out-2}:~} According to Eq. (\ref{eq:out_tf}),
\begin{align}
\bar{n}_{out}(t) & =\left\langle  \Psi_\nu ,\int_{t_0}^t g_{G^{+}}(t-\tau)^{\#}b(\tau)d\tau\left( \int_{t_0}^t g_{G^{-}}(t-\tau) b(\tau)d\tau\right)^{T}  \Psi_\nu\right\rangle \label{eq:nout7}
\\
&  +\left\langle  \Psi_\nu  , \int_{t_0}^t g_{G^{+}}(t-\tau)^{\#}b(\tau)d\tau\left(\int_{t_0}^t g_{G^{+}} (t-\tau) b^{\#}(\tau)d\tau\right)^{T}  \Psi_\nu\right\rangle   \nonumber
\\
& +\left\langle\Psi_\nu  , \int_{t_0}^t g_{G^{-}}(t-\tau)^{\#} b^{\#}(\tau)d\tau \left(  \int_{t_0}^t g_{G^{-}}(t-\tau) b(\tau)d\tau \right)^{T} \Psi_\nu\right\rangle  \nonumber
\\
&+\left\langle\Psi_\nu  , \int_{t_0}^t g_{G^{-}}(t-\tau)^{\#} b^{\#}(\tau)d\tau \left(\int_{t_0}^t g_{G^{+}} (t-\tau) b^{\#}(\tau)d\tau\right)^{T}  \Psi_\nu\right\rangle.  \nonumber
\end{align}
In what follows we evaluate each term of the right-hand side of Eq. (\ref{eq:nout7}). Firstly, it is easy to see that%
\begin{eqnarray}
\left\langle  \Psi_\nu ,\int_{t_0}^t g_{G^{+}}(t-\tau)^{\#}b(\tau)d\tau\left( \int_{t_0}^t g_{G^{-}}(t-\tau) b(\tau)d\tau\right)^{T}  \Psi_\nu\right\rangle &=&0, \label{eq:interm1}
\\
\left\langle\Psi_\nu  , \int_{t_0}^t g_{G^{-}}(t-\tau)^{\#} b^{\#}(\tau)d\tau \left(\int_{t_0}^t g_{G^{+}} (t-\tau) b^{\#}(\tau)d\tau\right)^{T}  \Psi_\nu\right\rangle &=& 0. \label{eq:interm1b}
\end{eqnarray}
Secondly, note that%
{\small
\begin{align}
& \left\langle\Psi_\nu  , \int_{t_0}^t g_{G^{+}}(t-\tau)^{\#}b(\tau)d\tau\left(\int_{t_0}^t g_{G^{+}} (t-\tau) b^{\#}(\tau)d\tau\right)^{T}  \Psi_\nu\right\rangle \label{eq:interm1c}
 \\
& =\left\langle\Psi_\nu, \left[
\begin{array}{c}
\sum_{j=1}^{m}\int_{t_0}^{t}g_{G^+}^{1j}(t-\iota)^\ast b_{j}(\iota
)d\iota \\
\vdots \\
\sum_{j=1}^{m}\int_{t_0}^{t}g_{G^+}^{mj}(t-\iota)^\ast b_{j}(\iota)d%
\iota%
\end{array}
\right] \left[
\begin{array}{ccc}
\sum_{j=1}^{m}\int_{t_0}^{t}g_{G^+}^{1j}(t-\iota)b_{j}^{\ast}(%
\iota)d\iota & \cdots & \sum_{j=1}^{m}\int_{t_0}^{t}g_{G^+}^{mj}(t-%
\iota)b_{j}^{\ast}(\iota)d\iota%
\end{array}
\right]\Psi_\nu\right\rangle . \nonumber
\end{align}
}
For given $k$ and $l$ ($k,l=1,\ldots,m$), we have
\begin{align}
& \left\langle  \Psi_\nu  ,\sum_{j=1}^{m}\int_{t_0}^{t}g_{G^+}^{kj}(t-r)^\ast b_{j}(r)dr\sum_{j=1}^{m}\int_{t_0}^{t}g_{G^+}^{lj}(t-r)b_{j}^{\ast}(r)dr \Psi_\nu  \right\rangle  \label{eq:interm1d} \\
& =\sum_{j=1}^{m}\int_{t_0}^{t}g_{G^+}^{kj}(t-r)^\ast g_{G^+}^{lj}(t-r)dr+\sum_{j=1}^{m}\int_{t_0}^t g_{G^+}^{kj}(t-r)^\ast \nu_j(r)dr \int_{t_0}^t g_{G^+}^{lj}(t-r) \nu_j(r)^\ast dr. \nonumber
\end{align}
Consequently, Sending $t_0 \to -\infty$ gives
\begin{align}
& \left\langle  \Psi_\nu  , \int_{t_0}^t g_{G^{+}}(t-\tau)^{\#}b(\tau)d\tau\left(\int_{t_0}^t g_{G^{+}} (t-\tau) b^{\#}(\tau)d\tau\right)^{T}  \Psi_\nu\right\rangle  \label{eq:interm3}
\\
=&\int_0^\infty g_{G^{+}}(r)^{\#}g_{G^{+}}(r)^{T}dr+\xi_{out}^{+}(t)^{\#}\xi_{out}^{+}(t)^{T}.  \nonumber
\end{align}
Thirdly, note that%
\begin{align}
& \left\langle\Psi_\nu  , \int_{t_0}^t g_{G^{-}}(t-\tau)^{\#} b^{\#}(\tau)d\tau \left(  \int_{t_0}^t g_{G^{-}}(t-\tau) b(\tau)d\tau \right)^{T} \Psi_\nu\right\rangle  \\
& =\left[
\begin{array}{c}
\sum_{j=1}^{m}\int_{t_0}^{t}g_{G^-}^{1j}(t-\iota)^\ast b_{j}^{\ast}(\iota)d%
\iota \\
\vdots \\
\sum_{j=1}^{m}\int_{t_0}^{t}g_{G^-}^{mj}(t-\iota)^{\ast}b_{j}(\iota)d%
\iota%
\end{array}
\right] \left[
\begin{array}{ccc}
\sum_{j=1}^{m}\int_{t_0}^{t}g_{G^-}^{1j}(t-\iota)b_{j}(\iota)d\iota & \cdots
& \sum_{j=1}^{m}\int_{t_0}^{t}g_{G^-}^{mj}(t-\iota)b_{j}(\iota)d\iota%
\end{array}
\right]. \nonumber
\end{align}
For given $k$ and $l$ ($k,l=1,\ldots,m$), we have
\begin{align}
&\left\langle  \Psi_\nu  ,\left(\sum_{j=1}^{m}\int_{t_0}^{t}g_{G^-}^{kj}(t-r)^\ast b_{j}^{\ast}(r)dr\right) \left( \sum_{j=1}^{m}\int_{t_0}^{t}g_{G^-}^{lj}(t-r)b_{j}(r)dr\right) \Psi_\nu\right\rangle \nonumber
\\
=& \sum_{j=1}^{m}\int_{t_0}^t g_{G^-}^{kj}(t-r)^\ast \nu_j(r)^\ast dr \int_{t_0}^t g_{G^-}^{lj}(t-r) \nu_j(r) dr.  \label{eq:interm1e}
\end{align}
As a result, sending $t_0 \to -\infty$ gives
\begin{equation}
\left\langle\Psi_\nu  , \int_{t_0}^t g_{G^{-}}(t-\tau)^{\#} b^{\#}(\tau)d\tau \left(  \int_{t_0}^t g_{G^{-}}(t-\tau) b(\tau)d\tau \right)^{T} \Psi_\nu\right\rangle
=\xi_{out}^{-}(t)^{\#}\xi_{out}^{-}(t)^{T}.  \label{eq:interm4}
\end{equation}
Finally, substituting Eqs. (\ref{eq:interm1}), (\ref{eq:interm1b}), (\ref{eq:interm3}), and (\ref{eq:interm4}) into Eq. (\ref{eq:nout7}) yields Eq. (\ref{eq:out}). Eq. (\ref{eq:tr1}) follows immediately from Eq. (\ref{eq:out}). This completes the proof. \hfill $\blacksquare$

\emph{Proof of Theorem \ref{thm:main}:~} First we prove that the
steady-state output states are indeed in the form of Eq. (\ref{class_F}). In
analog to Eq. (\ref{rho_infty}), the steady-state joint system and output field state is
\begin{eqnarray}
\rho_\infty &=&\lim_{t\rightarrow \infty ,t_{0}\rightarrow -\infty}U(t,t_{0})\left\vert \phi \right\rangle \left\langle\phi\right\vert \otimes \rho_{\xi _{in},R_{in}}U(t,t_{0})^\ast  \label{eq:temp_1} \\
&=&\left(I_{sys}\otimes\prod_{k=1}^{m}\sum_{j=1}^{m}\int_{-\infty }^{\infty }\left\{
 \left( \xi _{in}^{-,jk}(r)b_{j}^{-\ast }(r,-\infty )dr -\xi _{in}^{+,jk}(r)^{\ast }b_{j}^{-}(r,-\infty
)\right) dr\right\}\right)\rho_{\infty,g}  \notag \\
&&\times \left(I_{sys}\otimes \prod_{k=1}^{m}\sum_{j=1}^{m}\int_{-\infty }^{\infty
}\left\{\left( \xi _{in}^{-,jk}(r)b_{j}^{-\ast }(r,-\infty
)dr -\xi _{in}^{+,jk}(r)^{\ast
}b_{j}^{-}(r,-\infty )\right) dr\right\} \right) ^{\ast }  \notag
\end{eqnarray}
where
\begin{equation}\label{eq:rho_infty_g}
 \rho_{\infty,g} =  \lim_{t\rightarrow \infty ,t_{0}\rightarrow -\infty}U(t,t_{0})\left\vert \phi \right\rangle \left\langle \phi \right\vert\otimes \rho _{R_{in}}U(t,t_{0})^\ast,
\end{equation}
and for $j=1,\ldots ,m$,
\begin{equation}
b_{j}^{-}(t,-\infty )=\lim_{t_{0}\rightarrow -\infty
}U(t,t_{0})b_{j}(t)U^{\ast }(t,t_{0}),~b_{j}^{-\ast }(t,-\infty
)=\lim_{t_{0}\rightarrow -\infty }U(t,t_{0})b_{j}^{\ast }(t)U^{\ast
}(t,t_{0}).  \label{eq:bj-bj-ast}
\end{equation}
According to Eq. (\ref{eq:out_tf}), for any $t\in \mathbb{R}$,
\begin{equation}
\breve{b}(t)=\int_{t_{0}}^{\infty }g_{G}(t-r)U(r,t_{0})\breve{b}%
(r)U(r,t_{0})^{\ast }dr+e^{A(t-t_{0})}U(t,t_{0})\breve{a}U(t,t_{0})^{\ast }.
\label{eq:bt}
\end{equation}%
Letting $t_{0}\rightarrow -\infty $ and substituting Eq. (\ref{eq:bj-bj-ast}%
) into Eq. (\ref{eq:bt}) yield $\breve{b}(t)=\int_{-\infty }^{\infty
}g_{G}(t-r)\breve{b}^{-}(r,-\infty )dr$. As a result
\begin{equation}
\breve{b}^{-}(t,-\infty )=\int_{-\infty }^{\infty }g_{G^{-1}}(t-r)\breve{b}%
(r)dr.  \label{eq:bUbG}
\end{equation}%
Partition $g_{G^{-}}$ and $g_{G^{+}}$ as
\begin{equation}
g_{G^{-}}(t)=\left[
\begin{array}{ccc}
g_{G^{-}}^{1}(t) & \cdots  & g_{G^{-}}^{m}(t)%
\end{array}%
\right] ,\ \ g_{G^{+}}(t)=\left[
\begin{array}{ccc}
g_{G^{+}}^{1}(t) & \cdots  & g_{G^{+}}^{m}(t)%
\end{array}%
\right] .
\end{equation}%
According to Lemma \ref{lem:G_inv}, for $j=1,\ldots ,m$,
\begin{equation}
b_{j}^{-\ast }(r,-\infty )=-\int_{-\infty }^{\infty }g_{G^{+}}^{j}(\iota
-r)^{\dag }b(\iota )d\iota +\int_{-\infty }^{\infty }g_{G^{-}}^{j}(\iota
-r)^{T}b^{\#}(\iota )d\iota . \label{eq:G-inv2}
\end{equation}%
In a similar way,
\begin{equation}
b_{j}^{-}(r,-\infty )=\int_{-\infty }^{\infty }g_{G^{-}}^{j}(\iota -r)^{\dag
}b(\iota )d\iota -\int_{-\infty }^{\infty }g_{G^{+}}^{j}(\iota
-r)^{T}b^{\#}(\iota )d\iota .  \label{eq:G-inv3}
\end{equation}
By Eqs. (\ref{eq:G-inv2}) and (\ref{eq:G-inv3}), for each $j,k=1,\ldots,m$,
\begin{eqnarray}
&&\int_{-\infty }^{\infty }\xi _{in}^{-,jk}(r)b_{j}^{-\ast }(r,-\infty
)dr|\phi \Phi _{R_{out}}\rangle -\xi _{in}^{+,jk}(r)^{\ast
}b_{j}^{-}(r,-\infty )dr  \label{Eq:temp} \\
&=&-\int_{-\infty }^{\infty }\xi _{in}^{-,jk}(r)\int_{-\infty }^{\infty
}g_{G^{+}}^{j}(\iota -r)^{\dag }b(\iota )d\iota dr+\int_{-\infty }^{\infty
}\xi _{in}^{-,jk}(r)\int_{-\infty }^{\infty }g_{G^{-}}^{j}(\iota
-r)^{T}b^{\#}(\iota )d\iota dr  \notag \\
&&-\sum_{j=1}^{m}\int_{-\infty }^{\infty }\xi _{in}^{+,jk}(r)^{\ast
}\int_{-\infty }^{\infty }g_{G^{-}}^{j}(\iota -r)^{\dag }b(\iota )d\iota
dr+\int_{-\infty }^{\infty }\xi _{in}^{+,jk}(r)^{\ast }\int_{-\infty
}^{\infty }g_{G^{+}}^{j}(\iota -r)^{T}b^{\#}(\iota )d\iota dr  \notag \\
&=&B_{j}^{\ast }(\xi _{out}^{-,jk})-B_{j}(\xi _{out}^{+,jk}).  \notag
\end{eqnarray}%
Substituting Eq. (\ref{Eq:temp}) into Eq.  (\ref{eq:temp_1}) gives
\begin{equation*}
\rho_\infty = \left(I_{sys}\otimes\prod_{k=1}^{m}\sum_{j=1}^{m} (B_{j}^{\ast }(\xi
_{out}^{-,jk})-B_{j}(\xi _{out}^{+,jk}))\right)\rho_{\infty,g}\left(I_{sys}\otimes
\prod_{k=1}^{m}\sum_{j=1}^{m} (B_{j}^{\ast }(\xi
_{out}^{-,jk})-B_{j}(\xi _{out}^{+,jk}))\right) ^{\ast }.
\end{equation*}%
Thus%
\begin{equation}
\rho_{\xi_{out},R_{out}}=\mathrm{Tr}_{sys}\left[\rho_\infty\right]
=\prod_{k=1}^{m}\sum_{j=1}^{m}(B_{j}^{\ast }(\xi _{out}^{-,jk})-B_{j}(\xi
_{out}^{+,jk}))\rho _{R_{out}}\left(
\prod_{k=1}^{m}\sum_{j=1}^{m}(B_{j}^{\ast }(\xi _{out}^{-,jk})-B_{j}(\xi
_{out}^{+,jk}))\right) ^{\ast },  \label{Eq:temp_2}
\end{equation}%
where $\rho _{R_{out}}=\mathrm{Tr}_{sys}[\rho_{\infty,g}]$, whose covariance function $R_{out}$ is given in Eq. (\ref{eq:R_out_gnr}). Consequently, $\rho _{\xi _{out},R_{out}}$ is exactly in the form of Eq. (%
\ref{class_F}).

Next we prove that $\rho _{\xi _{out},R_{out}}$ is normalized. Noticing that
for a complex-valued function $\eta (t)\in \mathbb{C}$,
\begin{equation}
\int_{-\infty }^{\infty }\eta (t)b^{\ast }(t)dt=\int_{-\infty }^{\infty
}\eta \lbrack -i\omega ]b^{\ast }[i\omega ]d\omega ,\ \ \int_{-\infty
}^{\infty }\eta (t)^{\ast }b(t)dt=\int_{-\infty }^{\infty }\eta \lbrack
i\omega ]^{\ast }b[i\omega ]d\omega .  \label{eq:xib}
\end{equation}%
It is easy to show that the frequency counterpart of $\rho _{\xi
_{in},R_{in}}$ in Eq. (\ref{eq:psi_in_7}) is
\begin{eqnarray*}
\rho _{\xi _{in},R_{in}}&=&\underbrace{\int_{-\infty }^{\infty }\cdots
\int_{-\infty }^{\infty }}_{m}M_{\xi }^{-}(\omega _{1\rightarrow
m})^{T}\Theta ^{\otimes _{c}^{m}}M_{\breve{b}}(\omega _{1\rightarrow
m})d\omega _{1\rightarrow m}\rho _{R_{in}} \\
&& \ \ \times \underbrace{\int_{-\infty
}^{\infty }\cdots \int_{-\infty }^{\infty }}_{m}M_{\xi }^{+}(\omega
_{1\rightarrow m})^{\dag }J^{\otimes _{c}^{m}}M_{\breve{b}}(\omega
_{1\rightarrow m})d\omega _{1\rightarrow m},
\end{eqnarray*}%
where the shorthand $d\omega _{1\rightarrow j}$ is used to denote $d\omega
_{1}\cdots d\omega _{j}$ for an arbitrary positive integer $j$.
Consequently, similar to Eq. (\ref{eq:psi_in_7}), the normalization
condition (\ref{eq:innerproduct}) is equivalent to
\begin{equation}
\underbrace{\int_{-\infty }^{\infty }\cdots \int_{-\infty }^{\infty }}%
_{2m}(M_{\xi }^{+}(\omega _{1\rightarrow m})^{\#}\otimes _{c}M_{\xi
}^{-}(\omega _{m+1\rightarrow 2m}))^{T}J^{\otimes _{c}^{m}}\otimes
_{c}\Theta ^{\otimes _{c}^{m}}\mathrm{Tr}[\rho _{R_{in}}M_{\breve{b}}(\omega
_{1\rightarrow 2m})]d\omega _{1\rightarrow 2m}=1.  \label{eq:innerproduct_b}
\end{equation}%
As a result, it suffices to show that $\rho _{\xi _{out},R_{out}}$ has the
normalization condition (\ref{eq:innerproduct_b}). Firstly, by Eq. (\ref%
{eq:xi_out}), we have
\begin{equation*}
\xi _{out}[i\omega ]=\left[
\begin{array}{cc}
\xi _{out}^{-}[i\omega ] & \xi _{out}^{+}[i\omega ] \\
\xi _{out}^{+}[-i\omega ]^{\#} & \xi _{out}^{-}[-i\omega ]^{\#}%
\end{array}%
\right] =\Xi _{G}[i\omega ]\xi \lbrack i\omega ].
\end{equation*}%
Partition $\xi _{out}$ to be $\xi _{out}=[\xi _{out}^{1}\ \cdots \ \xi
_{out}^{2m}]$. Then by Eq. (\ref{Eq:temp_2}),
\begin{eqnarray}
\mathrm{Tr}[\rho _{\xi _{out},R_{out}}] &=&\underbrace{\int_{-\infty }^{\infty
}\cdots \int_{-\infty }^{\infty }}_{2m}(M_{\xi _{out}}^{+}(\omega
_{1\rightarrow m})^{\#}\otimes _{c}M_{\xi _{out}}^{-}(\omega
_{m+1\rightarrow 2m}))^{T}  \label{eq:innerproduct_2} \\
&&\times J^{\otimes _{c}^{m}}\otimes _{c}\Theta ^{\otimes _{c}^{m}}\mathrm{Tr}[\rho
_{R_{out}}M_{\breve{b}}(\omega _{1\rightarrow 2m})]d\omega _{1\rightarrow
2m}.  \notag
\end{eqnarray}%
Secondly, noticing that $\xi _{out}^{k}[i\omega ]=\Xi _{G}[i\omega ]\xi
_{in}^{k}[i\omega ]$, Eq. (\ref{eq:innerproduct_2}) becomes
\begin{align}
\mathrm{Tr}[\rho _{\xi _{out},R_{out}}]& =\underbrace{\int_{-\infty }^{\infty }\cdots
\int_{-\infty }^{\infty }}_{2m}(M_{\xi _{in}}^{+}(\omega _{1\rightarrow
m})^{\#}\otimes _{c}M_{\xi _{in}}^{-}(\omega _{m+1\rightarrow 2m}))^{T}
\label{eq:innerproduct_3} \\
& \times (\Xi _{G}[i\omega ]^{\dag }J)^{\otimes _{c}^{m}}\otimes _{c}(\Xi
_{G}[-i\omega ]^{T}\Theta )^{\otimes _{c}^{m}}\mathrm{Tr}[\rho _{R_{out}}M_{\breve{b}%
}(\omega _{1\rightarrow 2m})]d\omega _{1\rightarrow 2m},  \notag
\end{align}%
where slight abuse of notation is used, that is,
\begin{equation}
(\Xi _{G}[i\omega ]^{\dag }J)^{\otimes _{c}^{m}}:=(\Xi _{G}[i\omega
_{1}]^{\dag }J)\otimes _{c}\cdots \otimes _{c}(\Xi _{G}[i\omega _{m}]^{\dag
}J),
\end{equation}%
and
\begin{equation}
(\Xi _{G}[-i\omega ]^{T}\Theta )^{\otimes _{c}^{m}}:=(\Xi _{G}[-i\omega
_{m+1}]^{T}\Theta )\otimes _{c}\cdots \otimes _{c}(\Xi _{G}[-i\omega
_{2m}]^{T}\Theta ).
\end{equation}%
Thirdly, denote
\begin{equation}
\Xi _{G}[i\omega ]^{\otimes _{c}^{2m}}:=\Xi _{G}[i\omega _{1}]\otimes
_{c}\cdots \otimes _{c}\Xi _{G}[i\omega _{2m}].
\end{equation}%
Noticing
\begin{equation}
\mathrm{Tr}[\rho _{R_{out}}M_{\breve{b}}(\omega _{1\rightarrow 2m})]=\Xi _{G}[i\omega
]^{\otimes _{c}^{2m}}\mathrm{Tr}[\rho _{R_{in}}M_{\breve{b}}(\omega _{1\rightarrow
2m})],
\end{equation}%
Eq. (\ref{eq:innerproduct_3}) becomes
\begin{align}
& \mathrm{Tr}[\rho _{\xi _{out},R_{out}}]  \label{eq:innerproduct_4} \\
=& \underbrace{\int_{-\infty }^{\infty }\cdots \int_{-\infty }^{\infty }}%
_{2m}(M_{\xi _{in}}^{+}(\omega _{1\rightarrow m})^{\#}\otimes _{c}M_{\xi
_{in}}^{-}(\omega _{m+1\rightarrow 2m}))^{T}  \notag \\
& \times (\Xi _{G}[i\omega ]^{\dag }J\Xi _{G}[i\omega ])^{\otimes
_{c}^{m}}\otimes _{c}(\Xi _{G}[-i\omega ]^{T}\Theta \Xi _{G}[i\omega
])^{\otimes _{c}^{m}}\mathrm{Tr}[\rho _{R_{in}}M_{\breve{b}}(\omega _{1\rightarrow
2m})]d\omega _{1\rightarrow 2m},  \notag
\end{align}%
where slight abuse of notation is used, that is,
\begin{equation}
(\Xi _{G}[i\omega ]^{\dag }J\Xi _{G}[i\omega ])^{\otimes _{c}^{m}}:=(\Xi
_{G}[i\omega _{1}]^{\dag }JG[i\omega _{1}])\otimes _{c}\cdots \otimes
_{c}(\Xi _{G}[i\omega _{m}]^{\dag }J\Xi _{G}[i\omega _{m}]),
\end{equation}%
and
\begin{equation}
(\Xi _{G}[-i\omega ]^{T}\Theta \Xi _{G}[i\omega ])^{\otimes _{c}^{m}}:=(\Xi
_{G}[-i\omega _{m+1}]^{T}\Theta \Xi _{G}[i\omega _{m+1}])\otimes _{c}\cdots
\otimes _{c}(\Xi _{G}[-i\omega _{2m}]^{T}\Theta \Xi _{G}[i\omega _{2m}].
\end{equation}%
Finally, Eq. (\ref{eq:innerproduct_4}), together with the relations
\begin{equation}
\Xi _{G}[i\omega _{k}]^{\dag }J\Xi _{G}[i\omega _{k}]=J,\ \ \Xi
_{G}[-i\omega _{k}]^{T}\Theta \Xi _{G}[i\omega _{k}]=\Theta ,\ \ (k=1,\ldots
,2m)
\end{equation}%
establishes Eq.(\ref{eq:innerproduct_b}). The proof is completed. \hfill $%
\blacksquare $


\end{document}